\documentclass{lmcs}

\usepackage{bookmark}   
\usepackage[UKenglish]{babel}
\usepackage{colorprofiles}
\usepackage[doipre={doi:}]{uri}   
\hypersetup{hidelinks}

\usepackage{csquotes}

\usepackage{import}

\usepackage{etoolbox}
\apptocmd{\sloppy}{\hbadness 10000\relax}{}{} 

\usepackage{mathtools}
\usepackage{physics}
\usepackage{amsfonts}
\usepackage{mathdots}   
\usepackage{mathrsfs}   
\usepackage{bigints}    
\usepackage{xfrac}      
\usepackage{stmaryrd}   
\SetSymbolFont{stmry}{bold}{U}{stmry}{m}{n}   

\usepackage{thm-restate}

\usepackage{placeins}

\usepackage{bm}   

\usepackage{blkarray}   

\usepackage{graphicx}
\graphicspath{figures}

\usepackage{tikzfig}
\usepackage{pgfplots} 
\pgfplotsset{compat=newest, compat/show suggested version=false} 

\tikzstyle{box}=[draw=black, shape=rectangle, fill=white, minimum size=.95em, inner sep=0.15em, scale=0.85, font={\scriptsize}]
\tikzstyle{gbox}=[box, draw=black, shape=rectangle, fill={zx_green}, tikzit fill={rgb,255: red,181; green,215; blue,181}, tikzit shape=rectangle]
\tikzstyle{hbox}=[box, draw=black, shape=rectangle, fill=yellow, minimum size=.55em]
\tikzstyle{gn}=[draw=black, shape=circle, fill={zx_green}, inner sep=0.7mm, minimum width=0pt, minimum height=0pt, tikzit fill={rgb,255: red,181; green,215; blue,181}]
\tikzstyle{gn_phase}=[shape=rectangle, fill={zx_green}, draw=black, minimum size=1em, rounded corners=0.4em, inner sep=0.2em, outer sep=-0.2em, scale=0.8, font={\scriptsize}, tikzit shape=circle, tikzit fill={rgb,255: red,130; green,188; blue,130}]
\tikzstyle{rn}=[gn, fill={zx_red}, draw=black, tikzit fill={rgb,255: red,255; green,165; blue,165}]
\tikzstyle{rn_phase}=[{gn_phase}, fill={zx_red}, draw=black, tikzit fill={rgb,255: red,215; green,96; blue,96}]
\tikzstyle{pn}=[gn, fill={zx_pink}, draw=black, tikzit fill=pink]
\tikzstyle{pn_phase}=[{gn_phase}, fill={zx_pink}, draw=black, tikzit fill=red]
\tikzstyle{rtriang}=[shape=isosceles triangle, fill=yellow, draw=black, isosceles triangle stretches=true, inner sep=0.8pt, minimum width=0.25cm, minimum height=2mm]
\tikzstyle{ltriang}=[rtriang, shape=isosceles triangle, fill=yellow, draw=black, shape border rotate=180]
\tikzstyle{utriang}=[rtriang, shape=isosceles triangle, fill=yellow, draw=black, shape border rotate=90]
\tikzstyle{dtriang}=[rtriang, shape=isosceles triangle, fill=yellow, draw=black, shape border rotate=-90]
\tikzstyle{lmat}=[shape=signal, signal to=west, signal from=east, fill={zx_grey}, draw=black, minimum height=6pt, inner sep=.75pt, font={\scriptsize \boldmath}, tikzit fill=gray, tikzit category=GLA]
\tikzstyle{rmat}=[lmat, shape=signal, signal to=east, signal from=west, tikzit fill=gray, tikzit category=GLA]
\tikzstyle{dmat}=[lmat, shape=signal, signal to=west, signal from=east, tikzit fill=gray, tikzit category=GLA, rotate=90]
\tikzstyle{umat}=[lmat, shape=signal, signal to=east, signal from=west, tikzit fill=gray, tikzit category=GLA, rotate=90]
\tikzstyle{uw}=[shape=isosceles triangle, isosceles triangle stretches=true, fill=black, draw=black, minimum width=2.8mm, minimum height=2mm, inner sep=1pt, outer sep=0pt, shape border rotate=90]
\tikzstyle{dw}=[uw, shape border rotate=-90, tikzit fill=black, tikzit draw=black]
\tikzstyle{d_split}=[shape=trapezium, fill=white, draw=black, inner sep=0pt, trapezium stretches body, text width=15pt, text height=7pt]
\tikzstyle{d_merge}=[{d_split}, shape=trapezium, draw=black, rotate=180]
\tikzstyle{sd_split}=[shape=trapezium, fill=white, draw=black, inner sep=0pt, trapezium stretches body, text width=10pt, text height=5pt]
\tikzstyle{sd_merge}=[{sd_split}, shape=trapezium, draw=black, rotate=180]
\tikzstyle{wire label}=[font={\tiny}, auto]
\tikzstyle{control}=[draw=black, shape=circle, fill=black, inner sep=0.5mm]

\tikzstyle{braceedge}=[-, decorate, decoration={brace, amplitude=2mm, raise=-1mm}]
\tikzstyle{dotsedge}=[-, dotted, decoration={brace, amplitude=2mm, raise=-1mm}]

\definecolor{zx_grey}{RGB}{211,211,211}
\definecolor{zx_red}{RGB}{232,165,165}
\definecolor{zx_pink}{RGB}{255, 130, 160}
\definecolor{zx_green}{RGB}{216,248,216}

\newcommand{\N}{\mathbb{N}}
\newcommand{\Z}{\mathbb{Z}}

\ifdef{\C}
  {\renewcommand{\C}{\mathbb{C}}}
  {\newcommand{\C}{\mathbb{C}}}
\newcommand{\minu}{\texttt{-}}
\newcommand{\plus}{\texttt{+}}

\newcommand{\interp}[1]{\left\llbracket#1\right\rrbracket}

\newcommand{\ZXWf}{\mathbf{ZXW_{f}}}
\newcommand{\ZXWd}{\mathbf{ZXW_{d}}}
\newcommand{\FHilb}{\textbf{FHilb}}
\newcommand{\lst}[1]{\left(#1\right)}

\newcommand{\bR}{\begin{color}{red}}
\newcommand{\bB}{\begin{color}{blue}}
\newcommand{\bM}{\begin{color}{magenta}}
\newcommand{\bC}{\begin{color}{cyan}}
\newcommand{\bW}{\begin{color}{white}}
\newcommand{\bBl}{\begin{color}{black}}
\newcommand{\bG}{\begin{color}{green}}
\newcommand{\bY}{\begin{color}{yellow}}
\newcommand{\e}{\end{color}}

\newcounter{axiom}
\newcommand{\axiom}[1]{(#1)\refstepcounter{axiom}\label{rule:#1}}
\newcommand{\tikzrefsize}[1]{\scriptsize{#1}}
\newcommand{\axref}[1]{\tikzrefsize{(\hyperref[rule:#1]{#1})}}
\newcommand{\textaxref}[1]{Axiom~(\hyperref[rule:#1]{#1})}
\newcommand{\textaxrefs}[2]{Axioms~(\hyperref[rule:#1]{#1}) and~(\hyperref[rule:#2]{#2})}
\newcommand{\lemref}[1]{\tikzrefsize{\eqref{#1}}}

\newcommand{\Mod}[1]{\ (\mathrm{mod}\ #1)}

\newcommand{\zxw}[1]{\cite[#1]{poorCompletenessArbitraryFinite2023}}

\newcommand{\dcol}[1]{}
\newcommand{\scol}[1]{#1}
\newcommand{\Description}[1]{}

\usetikzlibrary{external}
\begin{document}

\title[Completeness of qufinite ZXW calculus]{Completeness of qufinite ZXW calculus,\texorpdfstring{\\}{} a graphical language for finite-dimensional\texorpdfstring{\\}{} quantum theory}

\author[Q. Wang]{Quanlong Wang}[a]
\author[B. Poór]{Boldizsár Poór}[a,b]
\author[R.\@ A.\@ Shaikh]{Razin A.\@ Shaikh}[a,b]

\address{Quantinuum, 17 Beaumont Street, Oxford, OX1 2NA, United Kingdom}
\address{University of Oxford, Oxford, United Kingdom}

\begin{abstract}
  
Finite-dimensional quantum theory serves as the theoretical foundation for quantum information and computation.
Mathematically, it is formalized in the category \textbf{FHilb}, comprising all finite-dimensional Hilbert spaces and linear maps between them.
However, there has not been a graphical language for \textbf{FHilb} which is both universal and complete and thus incorporates a set of rules rich enough to derive any equality of the underlying formalism solely by rewriting.
In this paper, we introduce the qufinite ZXW calculus — a graphical language for reasoning about finite-dimensional quantum theory.
We set up a unique normal form to represent an arbitrary tensor and prove the completeness of this calculus by demonstrating that any qufinite ZXW diagram can be rewritten into its unique normal form.
This result implies the equivalence of the qufinite ZXW calculus and the category \textbf{FHilb}, leading to a purely diagrammatic framework for finite-dimensional quantum theory with the same reasoning power.
In addition, we identify several domains where the application of the qufinite ZXW calculus holds promise.
These domains include spin networks, interacting mixed-dimensional systems in quantum chemistry, quantum programming, high-level description of quantum algorithms, and mixed-dimensional quantum computing.
Our work paves the way for a comprehensive diagrammatic description of quantum physics, opening the doors of this area to the wider public.

\end{abstract}

\maketitle

\section{Introduction}\label{sec:introduction}

Finite-dimensional quantum theory is a fundamental theory of nature,
governing the behaviour of fundamental particles as well as complex interactions between such systems.
It serves as the backbone of various areas such as quantum chemistry and condensed matter physics as well as quantum information and computation~\cite{nielsenQuantumComputationQuantum2010,coeckePicturingQuantum2017}.
Mathematically, finite-dimensional quantum theory is formalized within the compact closed category $\FHilb$~\cite{abramskyCategoricalSemanticsQuantum2004,selingerFiniteDimensional2012}.
This category comprises all finite-dimensional Hilbert spaces and linear maps between them.\footnote{In this paper, we omit the $0$-dimensional Hilbert space, as it is not required in our formalism.}

Graphical languages~\cite{selingerSurveyGraphical2011} have gained substantial influence in quantum information processing~\cite{abramskyCategoricalQuantumMechanics2008, heunenCategoriesQuantumTheory2019}.
Among these, the ZX calculus~\cite{coeckeInteractingQuantumObservables2008, coeckeInteractingQuantumObservables2011} stands out as a prominent graphical language broadly applied in the field of quantum computation.
However, other calculi are also gaining popularity, such as the ZW~\cite{coeckeThreeQubitEntanglement2011, hadzihasanovicDiagrammaticAxiomatisationQubit2015},
ZH~\cite{backensZHCompleteGraphical2019, royQuditZHCalculusGeneralised2023},
and the ZXW calculus~\cite{shaikhHowSum2023, poorCompletenessArbitraryFinite2023, defeliceLightMatterInteractionZXW2023}.
Since their introduction, these calculi have significantly contributed to various domains,
including quantum circuit optimization~\cite{debeaudrapFastEffective2020, debeaudrapTechniquesReducePi2020, kissingerReducingTcountZXcalculus2020},
quantum error correction~\cite{debeaudrapZXCalculus2020, kissingerPhasefreeZXDiagrams2022, khesinGraphicalQuantumCliffordencoder2023},
measurement-based quantum computation~\cite{coeckeInteractingQuantumObservables2007, duncanGraphStatesNecessity2009, kissingerUniversalMBQC2019},
quantum natural language processing~\cite{coeckeFoundationsNearTermQuantum2020, meichanetzidisQuantumNaturalLanguage2021},
quantum machine learning~\cite{wangDifferentiatingIntegrating2024},
classical simulation~\cite{kissingerClassicalSimulationQuantum2022, codsiClassicallySimulatingQuantum2023, camSpeedingQuantumCircuits2023, kochContractionZX2024},
quantum foundations~\cite{coeckePhaseGroupsOrigin2011, backensCompleteGraphicalCalculus2016},
quantum chemistry~\cite{shaikhHowSum2023, defeliceLightMatterInteractionZXW2023},
complexity theory~\cite{debeaudrapTensorNetworkRewriting2021, laakkonenPicturingCountingReductions2023, laakkonenGraphicalSATAlgorithm2022},
quantum programming languages~\cite{caretteSZXCalculusScalable2019, borgnaEncodingHighlevelQuantum2023},
cognition~\cite{signorelliCompositionalModel2021},
and education in quantum computing~\cite{coeckeQuantumPictures2022, dundar-coeckeQuantumPicturalismLearning2023}.

Specific subcategories of $\FHilb$ are of particular interest in the study of quantum information processing.
Among them, $\FHilb_2$ is the most widely used subcategory, where Hilbert spaces have dimensions of $2^n$.
The states in this category are called \emph{qubits}, and all graphical languages mentioned above were originally developed for such systems.
However, graphical languages have since been extended to $\FHilb_3$, using \emph{qutrits}~\cite{wangQutritZXcalculusComplete2018}, $\FHilb_p$ where $p$ is an odd prime with \emph{quopits} as states~\cite{boothCompleteZXcalculi2022,poorQupitStabiliserZXtravaganza2023}, and $\FHilb_d$ for arbitrary \emph{qudits}~\cite{poorCompletenessArbitraryFinite2023}.
However, no language has previously been developed that can reason about the entirety of $\FHilb$.

In this paper, we introduce the qufinite ZXW calculus, a language that can perform any reasoning that can be done in $\FHilb$.
With this language we can address problems in a wide range of domains based on finite-dimensional quantum theory.
It enables the study of molecular interactions in quantum chemistry or explore spin networks for quantum gravity.
Another natural avenue is mixed-dimensional quantum computing~\cite{khanSynthesisMultiquditHybrid2006,baekkegaardRealizationEfficient2019} --- where circuits may have qudits of varying dimensions.
This approach has found applications in fields including the compression of quantum circuits~\cite{matoCompressionQubitCircuits2023}, more efficient synthesis of gates~\cite{lanyonSimplifyingQuantumLogic2009,diSynthesisMultivalued2013}, and native simulation of certain physical systems~\cite{banulsSimulatingLatticeGauge2020} such as nitrogen-vacancy center systems~\cite{bernienHeraldedEntanglement2013,awschalomQuantumSpintronics2013}.

Graphical calculi are assessed based on three fundamental properties, which are crucial for their effectiveness:
\begin{enumerate}
  \item \emph{Soundness}:
  The interpretation of any equality of diagrams is a valid equality of linear maps in $\FHilb$.
  In category theoretic terms, this interpretation from the graphical calculus category to its semantic category is a symmetric monoidal functor.

  \item \emph{Universality}:
  A robust graphical language should be capable of expressing every linear map within its framework.
  Category theoretically, this means that the interpretation functor is full.

  \item \emph{Completeness}:
  Any equation involving linear maps derivable in multilinear algebra should also be derivable within the graphical language through the process of rewriting.
  In other words, the interpretation functor is faithful.
\end{enumerate}
Of these properties, completeness poses the greatest challenge to prove.

Although the qubit ZX calculus was first formulated in 2007~\cite{coeckeInteractingQuantumObservables2007}, the crucial achievement of proving completeness was not realized until 2017~\cite{ngUniversalCompletionZXcalculus2017, jeandelDiagrammaticReasoningClifford2018}.
This achievement unfolded in numerous stages, each progressively expanding the fragment for which completeness was established~\cite{backensZXcalculusComplete2014, backensZXcalculusCompleteSinglequbit2014, jeandelCompleteAxiomatisationZXCalculus2018, hadzihasanovicTwoCompleteAxiomatisations2018}.
Smilarly, proofs of completenesses for qudit graphical calculi unfolded in several stages:
initially encompassing the stabilizer fragment of the qutrit ZX calculus~\cite{wangQutritZXcalculusComplete2018}, then the quopit ZX calculus with all odd prime dimensions~\cite{boothCompleteZXcalculi2022}, and most recently, the proof of completeness for the universal fragment of qudit ZXW calculus~\cite{poorCompletenessArbitraryFinite2023}.



\subsection{Our contributions}\label{subsec:our-contributions}

In this paper, we establish the framework of qufinite ZXW calculus whose category we denote as $\ZXWf$.
This new calculus incorporates the generators of the qudit ZXW calculi~\cite{poorCompletenessArbitraryFinite2023} and the dimension splitter (as well as its transpose, the dimension merger) from the qufinite ZX calculus~\cite{wangQufiniteZXcalculusUnified2022} to split (or combine) different dimensions.
Furthermore, we introduce a novel addition to our calculus, the \emph{mixed-dimensional Z-spider}, which enables the interaction of different dimensional wires.
Our set of rewrite rules includes that of the qudit ZXW calculus~\cite{poorCompletenessArbitraryFinite2023} with some rules generalized to the mixed-dimensional setting, along with two additional mixed-dimensional rules.

We establish a new normal form in the qufinite ZXW calculus, corresponding to an arbitrary tensor in \autoref{subsec:normal-form}.
With this normal form, we immediately obtain the universality of the calculus for multilinear algebra over complex numbers.

As one of our key results, we prove the completeness of the qufinite ZXW calculus in \autoref{completetm}.
We do this by demonstrating that any qufinite ZXW diagram can be rewritten into its normal form, since:
(1) all generators can be rewritten into their normal forms,
(2) the tensor product of any two normal forms can be rewritten into a single normal form,
and (3) a partially traced normal form can be rewritten into a normal form.

We develop a new technique (\autoref{lem:post-compose-mixed-z}) that enables the representation of mixed-dimensional quantum circuits embedded in a larger statespace.
This method proves to be a useful resource in many situations while proving completeness.
\begin{gather*}
    \input{figures/qufit-lemmas/multidim_diags_eq.tikz}
    \Longleftrightarrow
    \input{figures/qufit-lemmas/embedded_multidim_diags_eq.tikz}
\end{gather*}
On the other hand, this can allow a more efficient representation of certain qudit diagram as a mixed-dimensional quantum circuit.

Building on the completeness result, we prove that the category $\ZXWf$ is monoidally equivalent to the category $\FHilb$ in \autoref{catequivallence}.
This result implies that the diagrammatic formalism of $\ZXWf$ has the same reasoning power as $\FHilb$.
Therefore, any computation of $\FHilb$ now can be done solely with diagrammatic rewriting.

Finally, we explore potential applications of the qufinite ZXW calculus in \autoref{sec:apply}.
Using the quantum Fourier transformation as an example, we demonstrate that the calculus can express complicated quantum computations concisely, and could potentially serve as a high-level language for quantum computing.
We also show that the calculus is suitable for studying angular momentum coupling and spin networks by presenting a compact diagram for irreducible representations of $SU(2)$.
Furthermore, we argue that qufinite ZXW provides a valuable toolbox for reasoning about the Hamiltonians of interacting systems in quantum chemistry by representing the Jaynes-Cummings model diagrammatically

\section{Qufinite ZXW calculus}\label{sec:qufitzxw}

Graphical calculi are usually described in the framework of a strict monoidal category whose objects are spanned by one object.
These are called PROPs~\cite{bonchiInteractingHopfAlgebras2017} and are presented in terms of generators and rewrite rules.
For quantum computing, one usually interprets diagrams of a graphical calculus in a subcategory of $\FHilb$ called the semantic category.
However, to have a graphical calculus whose semantics cover the whole of $\FHilb$, the framework of PROP is not sufficient.
A solution to this problem is the use of a \emph{coloured PROP}~\cite{yauHigherDimensionalAlgebras2008,hackneyCategoryProps2015,caretteColoredPropsLarge2020} whose objects are finite lists of colours.

We introduce the qufinite ZXW calculus, whose semantics cover all of $\FHilb$.
This calculus is a coloured PROP denoted as $\ZXWf$ where the colours are positive integers, corresponding the dimension of each qubit.
In particular, we define the symmetric monoidal category $\ZXWf$ with objects as lists of dimensions $\lst{d_i}_{i=1}^n$ where $d_i \in \N \setminus \{0\}$ and morphisms generated by the diagrams,
for any $d, e \in \N \setminus \{0\}$ and $n, m \in \N \cup \{0\}$, and $\overrightarrow a \in \C^{d \minu 1}$:
\begin{align*}
  \input{qufit-generators/generalgreenspiderqdit2.tikz} &: \lst{d}^{\otimes n} \to \lst{d}^{\otimes m} & \qquad
  \begin{tikzpicture}
	\begin{pgfonlayer}{nodelayer}
		\node [style=hbox] (0) at (0, 0) {$H$};
		\node [style=none] (1) at (0, -1) {};
		\node [style=none] (2) at (0, 1) {};
	\end{pgfonlayer}
	\begin{pgfonlayer}{edgelayer}
		\draw (2.center) to (0);
		\draw (1.center) to (0);
	\end{pgfonlayer}
\end{tikzpicture}
 &: \lst{d} \to \lst{d} & \qquad
  \input{qufit-generators/w1to2.tikz} &: \lst{d} \to \lst{d,d} \\
  \input{qufit-generators/dimensionsplitter.tikz} &: \lst{de} \to \lst{d, e} & \qquad
  \input{qufit-generators/swap.tikz} &: \lst{d,e} \to \lst{e,d} & \qquad
  \begin{tikzpicture}
	\begin{pgfonlayer}{nodelayer}
		\node [style=none] (1) at (1, 0.85) {};
		\node [style=none] (2) at (1, -0.85) {};
		\node [style=none] (3) at (1, -1) {};
		\node [style=none] (4) at (1, 1) {};
	\end{pgfonlayer}
	\begin{pgfonlayer}{edgelayer}
		\draw (1.center) to (2.center);
	\end{pgfonlayer}
\end{tikzpicture}
 &: \lst{d} \to \lst{d}
\end{align*}
Diagrams are to be read top-to-bottom, as implied by the interpretation below.
Diagrams can be composed sequentially, by connecting input and output wires, and in parallel, by placing them side-by-side.
We also define the category of the qudit ZXW calculus as $\ZXWd$ in accordance with~\cite{poorCompletenessArbitraryFinite2023}, where objects are numbers (the number of wires) and morphisms are the above generators where only a single dimension appears in the diagram.

The qufinite ZXW calculus serves as a unified framework for all finite-dimensional qudit ZXW calculi~\cite{poorCompletenessArbitraryFinite2023}.
Working with mixed-dimensions, each wire is labelled with its dimension, and only wires with the same dimension can be connected.
Since a wire labelled with dimension $1$ is just an empty diagram, throughout this paper each wire label will represent an integer strictly bigger than~$1$.

In this section, we delve into the details of the qufinite ZXW calculus, providing a thorough understanding of its structure and capabilities.
We begin with an introduction to the generators of qudit ZXW calculus as it is presented in~\cite{poorCompletenessArbitraryFinite2023}, describing their interpretation as linear maps, and presenting useful notations.
Then, we introduce the mixed-dimensional generators of the language~\cite{wangQufiniteZXcalculusUnified2022}, including the swap, the dimension-split, and its transpose, the dimension-merge.
We also define the generalization of the Z-spider to mixed-dimensions.
Finally, we present a complete set of rules governing this calculus composed of the qudit ZXW rules inherited from~\cite{poorCompletenessArbitraryFinite2023}, mixed-dimensional generalizations of six qudit rules, and two new mixed-dimensional rules.

\subsection{Qudit generators and their interpretation\texorpdfstring{~\cite{poorCompletenessArbitraryFinite2023}}{}}\label{subsec:generators_n_interpretation}

This subsection assumes that each wire is labelled with a fixed positive integer $d$;
thus, for convenience, labels are omitted here.
The generators of the qudit ZXW calculus together with their standard interpretation $\interp{\cdot}$ are:
\begin{itemize}
  \item The \emph{Hadamard box},
  \[
    
    \quad \overset{\interp{\cdot}}{\longmapsto} \quad
    \frac{1}{\sqrt{d}}\sum_{k, j=0}^{d-1}\omega^{jk}\ket{j}\bra{k},
  \]
  where $\omega = e^{i\frac{2\pi}{d}}$ is the $d$-th root of unity.
  That is, it corresponds to the adjoint of the discrete Fourier transform matrix on d-dimensions.
  Notably, the qudit Hadamard box is not self-adjoint anymore as it is in the qubit case, so two Hadamard boxes do not equal the identity.

  \item The \emph{Z box},
  \[
    \input{qufit-generators/generalgreenspiderqdit2.tikz}
    \quad \overset{\interp{\cdot}}{\longmapsto} \quad
    \sum_{j=0}^{d-1}a_j\ket{j}^{\otimes m}\bra{j}^{\otimes n},
  \]
  where $a_0 = 1$ and $\overrightarrow{a}=(a_1, \cdots, a_{d-1})$ is an arbitrary complex vector, and we take the indices modulo $d$, that is, $a_j = a_{j\ \mathrm{mod}\ d}$ for $j \in \Z$.

  \item The \emph{W node},
  \[
    \input{qufit-generators/w1to2.tikz}
    \quad \overset{\interp{\cdot}}{\longmapsto} \quad
    \ket{00}\bra{0}+\sum_{i=1}^{d-1}(\ket{0i}+\ket{i0})\bra{i}.
  \]

  \item Lastly, the \emph{identity},
  \[
    
    \quad \overset{\interp{\cdot}}{\longmapsto} \quad
    I_d=\sum_{j=0}^{d-1}\ket{j}\bra{j}.
  \]
\end{itemize}

\subsection{Mixed-dimensional generators and their interpretation}

We introduce the generators whose wires have different dimensions, together with their standard interpretation $\interp{\cdot}$, as given in~\cite{wangQufiniteZXcalculusUnified2022}:
\begin{itemize}
  \item The \emph{swap},
  \[
    \input{qufit-generators/swap.tikz}
    \quad \overset{\interp{\cdot}}{\longmapsto} \quad
    \sum_{i=0}^{d-1}\sum_{j=0}^{e-1}\ket{j,i}\bra{i,j}.
  \]

  \item The \emph{dimension splitter},
  \[
    \input{qufit-generators/dimensionsplitter.tikz}
    \quad \overset{\interp{\cdot}}{\longmapsto} \quad
    \sum_{i=0}^{d-1}\sum_{j=0}^{e-1}\ket{i,j}\bra{in+j}.
  \]

\end{itemize}

\subsection{Notations}

\subsubsection{Qudit notations}

For convenience, we introduce the following notation which will be used throughout the paper.
This subsection assumes that each wire is labelled with a fixed positive integer $d$;
thus, for convenience, labels are omitted here.
\begin{itemize}
  \item For some $x \in \mathbb{C}$,
  \begin{gather*}
    \input{definitions/circlegspiders_2.tikz}
    \scol{\qquad \quad} \dcol{\\}
    \input{definitions/lastditgbox2.tikz}
    \scol{\qquad \quad} \dcol{\\}
    \input{definitions/firstditgbox.tikz}
  \end{gather*}
  Additionally, we often label the Z box with $\overrightarrow{1_k}$ where $\overrightarrow{1_k} = \overbrace{(\underbrace{1,\dotsc, 1}_{k - 1},0, \dotsc, 0)}^{d - 1}.$

  \item The qudit version of the Bell state is $\ket{00} + \ket{11} + \ket{22} + \cdots + \ket{(d-1) (d-1)}$.
  In the ZXW calculus, the Bell state and its transpose can be defined as, respectively:
  \begin{gather}
    \input{definitions/compactstructures_2.tikz}
    \dcol{\nonumber \\}\scol{\qquad \qquad}
    \input{definitions/compactstructures_1.tikz}
    \tag{S3}\label{rule:S3}\refstepcounter{equation}
  \end{gather}
  We refer to these diagrams as caps and cups, and structures equipped with them are called compact structures~\cite{coeckeInteractingQuantumObservables2011}.

  \item The $H$ box can be used to define the X-spider. For $j \in \mathbb{Z} / d\mathbb{Z}$,\\
  \begin{minipage}{\scol{.6}\linewidth}
    \begin{equation*}
      \input{definitions/pinkspiders.tikz}
    \end{equation*}
  \end{minipage}
  \begin{minipage}{\scol{.4}\linewidth}
    \[
      \input{definitions/pinkspiders-phasefree.tikz}
      \tag{HZ}\label{rule:HZ}\refstepcounter{equation}
    \]
  \end{minipage}
  where $K_j =\left(e^{i j\frac{2\pi}{d}}, e^{i 2j\frac{2\pi}{d}}, \cdots, e^{i (d-1)j\frac{2\pi}{d}}\right)$.
  Furthermore, $u_{m,n} = d^{\frac{m+n-2}{2}}-1$, hence the green box \begin{tikzpicture}
	\begin{pgfonlayer}{nodelayer}
		\node [style=gbox] (0) at (0, 0) {$u_{m,n}$};
	\end{pgfonlayer}
\end{tikzpicture}
 represents the scalar $d^{\frac{m+n-2}{2}}$.
  The $K_j$ vector phases are similar to the qubit $0$ and $\pi$ phases in a sense that these correspond to the phases of Pauli operators.
  If $d$ is the dimension, there are exactly $d$ such phases.

  It is useful to note that the interpretation of the X-spider for each $K_j$ phase is given as:
  \[
    \begin{tikzpicture}
	\begin{pgfonlayer}{nodelayer}
		\node [style=none] (0) at (0, 1.25) {$\cdots$};
		\node [style={rn_phase}] (1) at (0, 0) {$K_j$ };
		\node [style=none] (2) at (1, -1.25) {};
		\node [style=none] (3) at (-1, 1.25) {};
		\node [style=none] (4) at (0, -1.25) {$\cdots$};
		\node [style=none] (5) at (-1, -1.25) {};
		\node [style=none] (6) at (1, 1.25) {};
		\node [style=none] (7) at (-1.125, 1.875) {};
		\node [style=none] (8) at (1.125, 1.875) {};
		\node [style=none] (9) at (-1.125, -1.875) {};
		\node [style=none] (10) at (1.125, -1.875) {};
		\node [style=none] (11) at (0, 2.375) {{\scriptsize $n$}};
		\node [style=none] (12) at (0, -2.375) {{\scriptsize $m$}};
		\node [style=none] (13) at (1, 1.5) {};
		\node [style=none] (14) at (-1, 1.5) {};
		\node [style=none] (15) at (1, -1.5) {};
		\node [style=none] (16) at (-1, -1.5) {};
	\end{pgfonlayer}
	\begin{pgfonlayer}{edgelayer}
		\draw [in=-90, out=150] (1) to (3.center);
		\draw [in=-90, out=30] (1) to (6.center);
		\draw [in=90, out=-150] (1) to (5.center);
		\draw [in=90, out=-30] (1) to (2.center);
		\draw [style=braceedge] (10.center) to (9.center);
		\draw [style=braceedge] (7.center) to (8.center);
		\draw (13.center) to (6.center);
		\draw (14.center) to (3.center);
		\draw (16.center) to (5.center);
		\draw (15.center) to (2.center);
	\end{pgfonlayer}
\end{tikzpicture}

    \quad \overset{\interp{\cdot}}{\longmapsto} \hspace{-.5cm}
    \sum_{\substack{
      0 \leq i_1, \cdots, i_m,  j_1, \cdots, j_n \leq d-1 \\
      i_1+\cdots+ i_m+j \equiv j_1+\cdots +j_n \Mod{d}}
    }
    \scol{\hspace{-1.75cm}}\dcol{\hspace{-1.5cm}}
    \ket{i_1, \cdots, i_m}\bra{j_1, \cdots, j_n}.
  \]
  X-spiders with a $K_j$ phases that have a single input or output now correspond to a computational basis costates/states as follows:
  \[
    \begin{tikzpicture}
	\begin{pgfonlayer}{nodelayer}
		\node [style={rn_phase}] (0) at (0, -0.5) {$K_j$};
		\node [style=none] (1) at (0, 0.75) {};
	\end{pgfonlayer}
	\begin{pgfonlayer}{edgelayer}
		\draw (0) to (1.center);
	\end{pgfonlayer}
\end{tikzpicture}

    \quad \overset{\interp{\cdot}}{\longmapsto} \quad
    \bra{j}
    \qquad \qquad
    \begin{tikzpicture}
	\begin{pgfonlayer}{nodelayer}
		\node [style={rn_phase}] (0) at (0, 0.75) {$K_j$};
		\node [style=none] (1) at (0, -0.5) {};
	\end{pgfonlayer}
	\begin{pgfonlayer}{edgelayer}
		\draw (0) to (1.center);
	\end{pgfonlayer}
\end{tikzpicture}

    \quad \overset{\interp{\cdot}}{\longmapsto} \quad
    \ket{d-j}
  \]


  \item We use a yellow $D$ box to denote the \emph{dualiser} as defined in~\cite{coeckeInteractingQuantumObservables2011}:
  \begin{gather}
    \input{definitions/dualiser.tikz}
    \quad \overset{\interp{\cdot}}{\longmapsto} \quad
    \sum_{i = 0}^{d \minu 1} \ket{d-i} \bra{i}.
    \tag{Du}\label{rule:Du}\refstepcounter{equation}
  \end{gather}

  \item The general $W$ node \scalebox{.75}{\input{definitions/w1ton.tikz}} and its transpose \scalebox{.75}{\input{definitions/wnto1.tikz}} are defined as:
  \begin{gather}
    \input{definitions/mlegsblackspider.tikz}
    \tag{WN}\label{rule:WN}\refstepcounter{equation}
  \end{gather}
  with interpretation
  \begin{align*}
    \input{definitions/w1ton.tikz}
    \quad \overset{\interp{\cdot}}{\longmapsto} &\quad
    \ket{0\cdots0}\bra{0} \dcol{\\&}+ \sum_{i=1}^{d - 1}(\ket{i0\cdots 00}+\cdots +\ket{00\cdots 0i})\bra{i}.
  \end{align*}


  \item A multiplier~\cite{bonchiInteractingHopfAlgebras2017, caretteSZXCalculusScalable2019, boothCompleteZXcalculi2022} labelled by $m$ indicates the number of connections between green and pink nodes.
  Unlike in the qubit case, a green and a red spider can be connected with more than one wire.
  In fact, the Hopf law generalizes to $d$ connections (see~\zxw{Lemma 27}) for a red and green spider to disconnect, so $m$ can be labeled modulo $d$:
  \begin{gather}
    \input{definitions/multiplier.tikz}
    \qquad \qquad
    \input{lemmas/multipliers/multiplier-mod.tikz}
    \tag{Mu}\label{rule:Mu}\refstepcounter{equation}
    \scol{\qquad \qquad}\dcol{\\ \nonumber}
    \input{definitions/multiplier-t.tikz}
  \end{gather}
\end{itemize}

\subsubsection{Dimension splitter}

\begin{itemize}
  \item We define the dimension splitter with zero and one legs as follows:
  \begin{gather}
    \input{qufit-generators/edge-dimsplit.tikz}
    \tag{SD}\label{rule:SDT}\refstepcounter{equation}
  \end{gather}

  \item Due to the associativity of the dimension splitter (\autoref{rule:DAS}),
  we can define the multiple-legged dimension splitter  as follows:
  \begin{gather}
    \input{qufit-generators/mlegsdimsplitspider.tikz}
    \tag{SD}\label{rule:SD}\refstepcounter{equation}
    \quad \overset{\interp{\cdot}}{\longmapsto} \quad
    \sum_{i_1=0}^{m_1-1} \cdots \sum_{i_k=0}^{m_k-1} \ket{i_1,\cdots, i_k}\bra{\sum_{j = 1}^{k} i_j \Pi_j },
  \end{gather}
  where $ \Pi_j = \prod_{l = j + 1}^{k} m_l$ for $1 \leq j \leq k - 1,$ and $\Pi_k = 1.$

  \item The \emph{dimension merger} can be defined as the transpose of the dimension splitter:
   \begin{equation}
    \input{qufit-generators/mlegsdimcombinespider.tikz}
  \end{equation}
\end{itemize}

\subsubsection{Mixed-dimensional Z box}
We define the mixed-dimensional Z box that can have legs of varying dimensions.
In the qudit setting, a Z box behaves as the generalized Kronecker delta: it ensures that the same basis state is present on each of its legs.
We preserve this behaviour in the mixed-dimensional case by selecting the $k$-th standard basis on each leg for any $k$ less than the minimal dimension.
As this is not possible for basis states greater than the minimum dimension, we set their coefficients to $0$.
This interpretation is given as follows:
\[
  \input{qufit-generators/mixed-zbox.tikz}
  \quad \overset{\interp{\cdot}}{\longmapsto} \quad
  \sum_{j=0}^{\min{\{d_i\}_i} - 1} a_j
  \ket{j, \cdots, j} \bra{j, \cdots, j},
\]
where $a_0 \coloneqq 1$ and $\overrightarrow{a} = (a_1, \cdots, a_{\min{\{d_i\}_i}-1})$.
Note that when each leg has the same dimension, this interpretation agrees with the qudit case.
We can construct such mixed-dimensional Z-box from the generators of qufinite ZXW calculus as follows:
\[
  \input{qufit-generators/mixed-zbox.tikz}
  \quad \coloneqq \quad
  \input{qufit-generators/mixed-zbox-decomp.tikz}
\]
where $d= \prod_{i} d_i$, and $\overrightarrow{a}' = (a_1, \cdots, a_{\min{\{d_i\}_i}-1}, 0, \cdots, 0)$.
The green circle spider for mixed dimensions can be defined similarly to the qudit case.

\begin{rem}
  All the diagrams can be made from composing the generators in finite times.
  There are two types of compositions for any two diagrams $D_1$ and $D_2$: the parallel composition $D_1 \otimes D_2$ where $D_1$ on the left of $D_2$, and sequential composition  $D_1 \circ D_2$  where the types of the outputs of $D_2$ exactly match the types of the inputs of $D_1$ and $D_2$ is placed above $D_1$.
  Therefore, general diagrams can be interpreted in the category $\FHilb$ as $\interp{D_1 \otimes D_2} =\interp{D_1} \otimes \interp{D_2},  \interp{D_1 \circ D_2} =\interp{D_1} \circ \interp{D_2}$, especially the empty diagram is interpreted as $1$.
\end{rem}

\begingroup 

\allowdisplaybreaks

\subsection{Complete set of rules}

In this section, we show a set of rewrite rules for qufinite ZXW calculus, which we later proved to be complete in \autoref{sec:complete-proof}.
In addition to the rules given below, the generators Z and W also satisfy the certain symmetry relations.
In particular,  the Z spider is flex-symmetric, i.e.\@ we can exchange any legs of the Z spider:
\begin{equation}
  \input{qufit-generators/z-flexsymmetry.tikz}\label{eq:z-flexsymmetry}
\end{equation}
On the other hand, the W node is symmetric, allowing us to permute any of the output legs.
\begin{equation}
  \input{axioms/wsymetrydit.tikz}\label{rule:Sym}
\end{equation}
We split the rule set into qudit-only rules and mixed-dimensional rules.
The qudit rules is given in Figure~\ref{fig:qudit-rules-table}, where each wire is assumed to be of dimension $d$.
The mixed-dimensional rules are given in Figure~\ref{fig:qufit-rules-table}.

\begin{figure}[htb]
  \input{axioms/qudit-rules-table.tikz}
  \caption[Qudit rules]{Qudit part of the rules. Here, $\overrightarrow{a}=(a_{1}, \dotsc, a_{d-1})$, $\overleftarrow{a}=(a_{d-1}, \dotsc, a_1)$, $k_j(\overrightarrow{a})=\left(\frac{a_{1-j}}{a_{d-j}}, \dotsc, \frac{a_{d-1-j}}{a_{d-j}}\right)$, $e \neq 0 \mod d$ and $T_j=(\underbrace{0,\dotsc, 1}_{j}, \dotsc, 0)$ for $0 \leq j < d$.}
  \label{fig:qudit-rules-table}
\end{figure}

\begin{figure}[htb]
  \input{qufit-axioms/qufit-rules-table.tikz}
  \caption{Mixed-dimensional part of the rules. Here, $N = \min\{m, n_1, n_2\}$, $M = \min\left(\{m_t\}_{t=1}^j \cup \{n_t\}_{t=1}^\ell \cup \{r_t\}_{t=1}^s\right)$, $\overrightarrow{ab'}=(a_1 b_1, \dotsc, a_{M-1} b_{M-1}, 0, \dotsc , 0)$, and $e \neq 0 \mod n$.}
  \label{fig:qufit-rules-table}
\end{figure}

Note that we can always check the soundness of the axioms using the interpretation map $\interp{\cdot}$ to verify that the matrices are the same.
Alternatively, we found it convenient to check the soundness by examining all possible input basis states.
From the soundness of the qufinite ZXW calculus, we have the following

\begin{prop}
  \label{momoidalfunctor}
  The standard interpretation $\interp{\cdot}: \ZXWf \to \FHilb$ is a symmetric monoidal functor.
\end{prop}

\endgroup 

\section{Completeness}\label{sec:complete-proof}

The proof strategy for establishing the completeness of the qufinite ZXW calculus is analogous to that used for the qudit ZXW calculus~\cite{poorCompletenessArbitraryFinite2023}.
We first establish the completeness of the qudit fragment (diagrams when all the wires are of the same dimension) of the qufinite ZXW calculus.
Then, we define a unique normal form for qufinite ZXW diagrams.
We describe the strategy of proving completeness.
Finally, we prove that any qufinite ZXW diagram can be transformed into an equivalent normal form.
Completeness follows as a corollary of this, as two diagrams with the same interpretation will yield the same normal form.

\subsection{Qudit completeness}

This section establishes the completeness of the qudit fragment of the calculus, which is advantageous to leverage to streamline several proofs in the mixed-dimensional setting.
Consider two mixed-dimensional diagrams $D_1$ and $D_2$.
If we can find corresponding diagrams $D_1'$ and $D_2'$ such that
\begin{equation}
  \input{figures/qufit-lemmas/multidim_diags_eq.tikz}
  \qquad\Longleftrightarrow\qquad
  \input{figures/qufit-lemmas/multidim_diags_eq_qudit.tikz}
\end{equation}
then the qudit completeness result applied to $D_1' = D_2'$ implies that $D_1 = D_2$ holds as well.
In particular, we can show the following:
\begin{restatable}{prop}{quditImplied}
  For any two qudit diagram $D_1, D_2 \in \ZXWd$, if $\ZXWd \vdash D_1 = D_2$ then $\ZXWf \vdash \iota(D_1) = \iota(D_2)$, where $\iota : \ZXWd \hookrightarrow \ZXWf$ is the inclusion functor.
  In other words, the axioms of the qufinite ZXW calculus, depicted in \autoref{fig:qufit-rules-table} and \autoref{fig:qudit-rules-table}, implies the axioms of the qudit ZXW calculus of Ref.~\cite{poorCompletenessArbitraryFinite2023}.
\end{restatable}

Moreover, we can show that any mixed-dimensional diagram can be embedded into qudits.
Once this embedding is established, we simply need to rewrite the resulting diagrams so that all the wires are qudits of the same dimension.
While the rewriting process is not trivial, it will simplify many proofs required to prove completeness.
\begin{prop}
  \label{lem:post-compose-mixed-z}
  The following two equalities are equivalent:
  \[
    \input{figures/qufit-lemmas/multidim_diags_eq.tikz}
    \qquad\Longleftrightarrow\qquad
    \input{figures/qufit-lemmas/embedded_multidim_diags_eq.tikz}
  \]
  where $d$ is a common multiple of all $d_k$ for $k \in \{1, \cdots, n\}$.
\end{prop}
\begin{proof}
  First, $\Longrightarrow$  follows from post-composing the diagrams with the given mixed-dimensional Z-spiders.
  Then, $\Longleftarrow$ follows from:
  \[
    \input{figures/qufit-lemmas/embedded_multidim_diags_eq.tikz}
    \ \Longrightarrow\ %
    \input{figures/qufit-lemmas/embedded_split_multidim_diags_eq.tikz}
    \ \Longrightarrow\ %
    \input{figures/qufit-lemmas/multidim_diags_eq.tikz}
  \]
  where the last implication is due to \autoref{lem:id-mixed-z-split}.
\end{proof}

\subsection{Normal form}\label{subsec:normal-form}

This subsection describes the normal form of the calculus and proves related theorems.
To construct the normal form, we use the concept of a mixed radix numeral system which is a non-standard positional numeral system in which the numerical base varies from position to position, first formalized by Cantor~\cite{cantorUeberEinfachen1869}.
An overview of the concept relevant to our work can be found in~\cite{knuthPositionalNumber1997}.
\begin{defi}[Mixed radix basis]
  For $K \in \N$, a \emph{finite mixed radix basis} is a sequence of integers $\mathbf{b} = (b_i)_{0 < i < K}$ such that each $b_i \geq 1$.
  We also define the product sequence $\pi^\mathbf{b} = (\pi^\mathbf{b}_i)_{0 \leq i}$ as $\pi^\mathbf{b}_0 \coloneqq 1$ and $\pi^\mathbf{b}_i \coloneqq \prod^i_{j = 1} b_j$ for $i \geq 1$.
\end{defi}
\begin{prop}
  \label{prop:mixed-radix-digit}
  Let $\mathbf{b} = (b_i)_{0 < i < K}$ be a finite mixed radix basis for some $K \in \N$.
  Then, for any integer $0 \leq n < \pi^\mathbf{b}_K$, there exists a unique sequence of non-negative integers $(c_i)_{0 \leq i < K}$ (called digits) such that $c_i < b_{i + 1}$ and
  \[
    n = \sum_{i = 0}^{K - 1} c_i \pi_i^{\mathbb{b}} = c_0 + c_1 b_1 + c_2 b_1 b_2 + \dots + c_{K - 1} b_1 \cdots b_{K - 1}
  \]
\end{prop}

\begin{thm}[Universality]
  The ZXW-calculus is universal.
  More formally, given $\mathbf{m} = (m_i)_{0 \leq i < s}$ a list of integers such that $m_i \geq 1$ and $s \in \N$.
  Let $A \in \bigotimes_{i = 0}^{s - 1} \mathbb{C}^{m_{i}}$ be an arbitrary tensor.
  We can construct a diagram such that the type of the diagram is $\mathbf{m}$ and its interpretation is the above given tensor.
\end{thm}
\begin{proof}
  Note that while $\mathbf{m}$ above is used as the type for a diagram, it can also be used as finite mixed radix basis.
  Let $m \coloneqq \pi^\mathbf{m}_s$ and $\overrightarrow{a} \in \mathbb{C}^{m}$ be the state vector isomorphic to $A$.
  We divide this into two cases based on the value of $m$, when it is $1$, and when it is greater than $1$.

  To construct the diagram for $A$ when $m \geq 1$, we write $\overrightarrow{a}$ as the sum of the computational basis states in $\C^{m}$, $\overrightarrow{a} = \sum_{\ell=0}^{m-1} a_\ell \ket{\ell}$.
  We can then represent each index $\ell$ with $\mathbf{m}$.
  Let $(e_{\ell, i})_{0 \leq i < s}$ be the digits of this mixed radix representation of $\ell$.
  Then, we can represent $\overrightarrow{a}$ as follows:
  \[
    \overrightarrow{a}
    \ =\ %
    \sum_{\ell=0}^{m-1} a_\ell \ket{\ell}
    \ =\ %
    \sum_{\ell=0}^{m-1} a_\ell \ket{e_{\ell, s \minu 1} \,\cdots\, e_{\ell, j} \,\cdots\, e_{\ell, 0}}
  \]

  We now construct the diagram representing $\overrightarrow{a}$.
  First, in dimension $m$, plugging $\ket{m - 1}$ into a $W$-node with $m$ legs gives the state $\ket{m\minu 1, \ 0 \cdots 0} + \ket{0,\ m\minu 1 \cdots 0} + \cdots + \ket{0,\ 0 \cdots m\minu 1}$.
  Then, we construct the diagram that keeps $\ket{0}$ unchaged and maps $\ket{m\minu 1}$ to $a_i \ket{i}$ for each $i$, from where outputs in mixed radix basis is easily acquired:
  \begin{align*}
    \input{qufit-normal-form/qufitnormalform-map.tikz}
    \ \overset{\interp{\cdot}}{\longmapsto} \quad
    a_i \ket{i}\bra{m - 1}\ +\ \ket{0}\bra{0}
    \qquad
    \input{qufit-normal-form/qufitnormalform-mixing.tikz}
    \ \overset{\interp{\cdot}}{\longmapsto} \quad
    &a_i \ket{e_{i, s \minu 1} \,\cdots\, e_{i, j} \,\cdots\, e_{i, 0}}\bra{m - 1}\\[-25pt]
    &+ \ket{0 \,\cdots\, 0 \,\cdots\, 0}\bra{0}\\
  \end{align*}
  Note that the above two diagrams may be equated using \autoref{rule:MS} if we compose the left diagram with the appropriate dimension splitter.
  Finally, putting together the $W$ state, the above gadget for each element in the summand, and summing this state using $X$ spiders, we acquire the diagram representing the tensor $A$:
  \begin{align}
    \label{eq:nf}
    &\input{qufit-normal-form/qufitnormalform-full.tikz}
    \quad \overset{\interp{\cdot}}{\longmapsto} \quad
    \begin{pmatrix}
      a_0     \\
      \vdots  \\
      a_i     \\
      \vdots  \\
      a_{m-1} \\
    \end{pmatrix}
    \scol{\quad}\dcol{\\} \dcol{&\hspace{3cm}}= \quad
    \sum_{\ell=0}^{m-1} a_\ell \ket{e_{\ell, m \minu 1} \,\cdots\, e_{\ell, s} \,\cdots\, e_{\ell, 0}}
    \ = \ %
    \overrightarrow{a}
  \end{align}
  and the diagram has type $\mathbf{m}$.

  When $m = 1$, then it must be the case that $\mathbf{m} = ()$ or that $m_i = 1$ for all $0 \leq i < s$.
  Therefore, all outputs have dimension one and the coefficient is $a_0$.
  This is represented as the following diagram:
  \[
    \input{qufit-normal-form/qufitnormalform-scalar.tikz}
  \]
  where $s \geq 0$.
\end{proof}

\begin{cor}
  \label{universalitylm}
  The interpretation functor $\interp{\cdot}$ is full.
\end{cor}

\begin{defi}[Normal form]
  A ZXW diagram is in normal form if it has the structure illustrated in \autoref{eq:nf}.
\end{defi}

\begin{prop}
  The normal form for a given tensor $A$ and type $\mathbf{m}$ is unique.
\end{prop}
\begin{proof}
  Consider a diagram $D$ in normal form representing tensor $A$ of type $\mathbf{m}$ (equivalent to the vector $\overrightarrow{a}$).
  We show each component of $D$ is uniquely determined by $A$ and $\mathbf{m}$.
  \begin{enumerate}
    \item The initial state preparation layer (using the W-node) and the final summation layer (with X-spiders) are determined solely by the type $\mathbf{m}$.

    \item The multiplier connectivity to the X-spiders ($e_{*,*}$) are uniquely determined by $\mathbf{m}$ and the uniqueness of mixed radix digits (\autoref{prop:mixed-radix-digit}).

    \item The coefficients of the Z-boxes ($a_{*}$) are given by the entries of the tensor.
  \end{enumerate}
  To conclude, every part of the diagram is uniquely determined for a given tensor $A$ and type $
  \mathbf{m}$, and therefore, the normal form is unique.
\end{proof}

If the output dimensions are the same, $m_{s-1}=\cdots = m_i = \cdots = m_0 \eqqcolon d$, then it is equivalent the qudit normal form.
This implies that any diagram in a qudit normal form can be converted into a qufinite normal form.
\begin{restatable}{lem}{quditqufitequivalence}
  \label{lem:qudit-qufit-equivalence}
  If all the output wires of a qufinite normal have the same dimension, then it is equivalent to a qudit normal form of~\cite{poorCompletenessArbitraryFinite2023}.
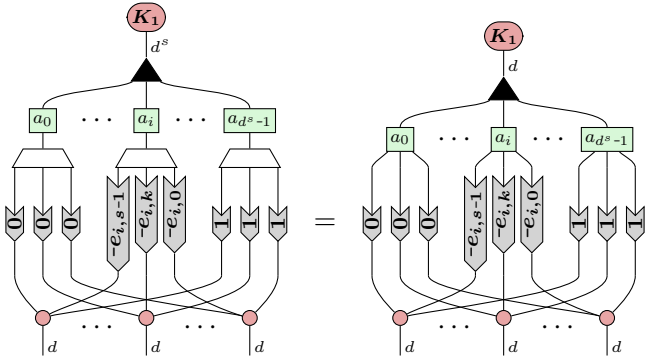
  \[
    \input{qufit-normal-form/qufitnormalformtrans.tikz}
  \]
\end{restatable}
The result follows from \autoref{k1dimchange} from the appendix along with~\eqref{rule:DZS} and~\eqref{rule:UT1}.
For a formal proof, see the appendix.

We now proceed to prove the completeness of the qufinite ZXW calculus.

\subsection{Proof strategy}
\newcommand{\proofstrategyscale}{0.8}

In this section, we describe the strategy for proving completeness.
The method is identical to the proof strategy of qudit ZXW calculus completeness~\cite{poorCompletenessArbitraryFinite2023}, but we recall it here to keep the paper self-contained.

To prove completeness, we need to show that for any two diagrams $A$ and $B$ with the same types, if $\interp{A} = \interp{B}$, then we can derive that $A = B$ from the rules of the calculus.
Due to the following map-state duality:
\begin{equation*}
  \scalebox{\proofstrategyscale}{
      \vspace{-2cm}
    \input{proof-idea/mapstatedual3.tikz}
  }
  \label{eq:maptostate}
\end{equation*}
we only need to consider state diagrams (i.e.\@ diagrams without any input).

Now assume that we have two state diagrams $A$ and $B$ with the same types such that $\interp{A} = \interp{B}$.
We need to show that $A = B$.
Since $\interp{A} = \interp{B}$,  $A$ and $B$ must have the same normal form $D$.
Therefore, if we can rewrite both $A$ and $B$ into $D$, i.e.\@ $A = D$ and $B = D$, then we obtain that $A = B$.
Thus, the proof strategy for completeness is to show that any state diagram can be rewritten into the unique normal form.

To rewrite an arbitrary state diagram into a normal form, we need to analyse its structure.
We first note that each state diagram $D_s$ has the following form:
\begin{equation*}
  \scalebox{\proofstrategyscale}{
    \input{proof-idea/generalstatedm.tikz}
  }
  \label{eq:generalstatediam}
\end{equation*}
where $A_1,\, A_2,\, \cdots,\, A_n$ are diagrams that are the parallel compositions of generators given in Section~\ref{sec:qufitzxw}.
We can rewrite $D_s$ into its normal form using the following steps:
\begin{enumerate}
  \item We rewrite $A_1$, which is the tensor product of generators without inputs, into its normal form $N_1$:
  \[
    \scalebox{\proofstrategyscale}{
      \input{proof-idea/a1_to_n1.tikz}
    }
  \]
  \item We bend the top of the diagram so that the generators in $A_2$ become top-bended state diagrams:
  \[
    \scalebox{\proofstrategyscale}{
      \input{proof-idea/sequentialbend.tikz}
    }
  \]
  \item We convert the top-bent $A_2$ diagram into its corresponding normal form $N_2$:
  \[
    \scalebox{\proofstrategyscale}{
      \vspace{-2cm}
      \input{proof-idea/a2_to_n2.tikz}
    }
  \]
  \item We rewrite the tensor product of the two normal forms into a single normal form denoted with $N_{1,2}$:
  \[
    \scalebox{\proofstrategyscale}{
      \input{proof-idea/n1_otimes_n2.tikz}
    }
  \]
  \item We reduce the partial traces of the normal form (i.e.\@ connections of two outputs of a normal form with a cup) into another normal form $N_{1,2}^\prime$:
  \[
    \scalebox{\proofstrategyscale}{
      \input{proof-idea/n12_partial_trace.tikz}
    }
  \]
  \item We repeat from step 2 for the rest of the diagram, that is from $A_3$ to $A_n$.
\end{enumerate}
If we follow the above steps, we obtain the normal form of $D_s$.
In summary, we need to prove the following to derive the completeness of the ZXW-calculus:
\begin{itemize}
  \item All generators bent in state diagrams or already being state diagrams can be rewritten into their normal forms.
  \item The tensor product of any two normal forms can be rewritten into a single normal form.
  \item A partial-traced normal form can be rewritten into a normal form.
\end{itemize}

\subsection{Proof of completeness}

\begin{restatable}[Completeness]{thm}{zxwcompleteness}
  \label{completetm}
  For finite-dimensional Hilbert spaces, the qufinite ZXW calculus is universally complete:
  For any two qufinite ZXW diagrams of the same type $D_1: A \to B$ and $D_2: A \to B$, if $\interp{D_1} = \interp{D_2}$, then $\ZXWf \vdash D_1 = D_2$.
\end{restatable}
\begin{proof}
  We divide this into two cases based on the values of $A$ and $B$:
  \begin{enumerate}
    \item $(\prod A_i) \cdot (\prod B_j) > 1$, that is, the diagrams are not scalars (\autoref{nonscalarcompletetm});
    \item $\prod A_i = 1$ and $\prod B_i = 1$, that is, the diagrams are scalars (\autoref{scalar-completetm}).
  \end{enumerate}
\end{proof}

\begin{cor}
  \label{catequivallence}
  The category $\ZXWf$ is monoidally equivalent to the category $\FHilb$.
\end{cor}
\begin{proof}
  Two categories are monoidally equivalent if there is a monoidal functor between them and the functor is full, faithful and essentially surjective on objects~\cite{heunenCategoriesQuantumTheory2019}.
  The interpretation functor $\interp{\cdot}$ is a monoidal functor by \autoref{momoidalfunctor}.
  It is full and faithful by \autoref{universalitylm} and \autoref{completetm} respectively.
  For any object $H \in \FHilb$, we have an object $[\dim(H)] \in \ZXWf$ such that $H \cong \mathbb{C}^{\dim(H)} = \interp{[\dim(H)]}$;
  hence, $\interp{\cdot}$ is essentially surjective on objects.
\end{proof}

\begin{restatable}[Non-Scalar Completeness]{prop}{zxwnonscalarcompleteness}
  \label{nonscalarcompletetm}
  For non-scalar finite-dimensional Hilbert spaces, the qufinite ZXW calculus is universally complete:
  For any two qufinite ZXW diagrams of the same type $D_1: A \to B$ and $D_2: A \to B$ where $(\prod A_i) \cdot (\prod B_j) > 1$, if $\interp{D_1} = \interp{D_2}$, then $\ZXWf \vdash D_1 = D_2$.
\end{restatable}
\begin{proof}
  To prove the above theorem, we first show that all the generators (the Z box, W node, Hadamard box, and the dimension splitter) can be rewritten into a normal form.
  Then, we show that any partial trace of a normal form can be transformed into another normal form.
  Finally, we rewrite the tensor product of two normal forms into a single normal form.
  The above statements are formalized in \autoref{lem:zbox}, \autoref{lem:wnode}, \autoref{lem:hadamard}, \autoref{lem:dimsplitter}, \autoref{lem:partialtrace} and \autoref{lem:tensor}.
  Note that the proofs of \autoref{lem:dimsplitter}, \autoref{lem:partialtrace} and \autoref{lem:tensor} are the content of the appendix.
\end{proof}

\begin{restatable}[Scalar Completeness]{prop}{zxwscalarcompleteness}
  \label{scalar-completetm}
  For scalars, the qufinite ZXW calculus is universally complete:
  For any two qufinite ZXW diagrams of the same type $D_1: A \to B$ and $D_2: A \to B$ where $\prod A_i = 1$ and $\prod B_i = 1$, if $\interp{D_1} = \interp{D_2}$, then $\ZXWf \vdash D_1 = D_2$.
\end{restatable}
\begin{proof}
  \[
    \begin{tikzpicture}
	\begin{pgfonlayer}{nodelayer}
		\node [style=none] (0) at (-4.25, 1) {};
		\node [style=none] (1) at (-4.25, 0) {};
		\node [style=none] (2) at (-1.25, 0) {};
		\node [style=none] (3) at (-1.25, 1) {};
		\node [style=none] (4) at (-2.75, 0.5) {$D$};
		\node [style=none] (5) at (-3.75, 0) {};
		\node [style=none] (6) at (-1.75, 0) {};
		\node [style=none] (7) at (-3.75, -0.75) {};
		\node [style=none] (8) at (-1.75, -0.75) {};
		\node [style=wire label] (9) at (-4, -0.5) {$1$};
		\node [style=wire label] (10) at (-1.5, -0.5) {$1$};
		\node [style=none] (11) at (-2.75, -0.5) {$\ldots$};
		\node [style=wire label] (12) at (-2.75, -1.5) {$0 \leq n$};
		\node [style=none] (13) at (-4, -1) {};
		\node [style=none] (14) at (-1.5, -1) {};
		\node [style=none, label={above:\lemref{rule:UT1}}] (15) at (0, 0) {$=$};
		\node [style=none] (57) at (6.75, 3.5) {};
		\node [style=none] (58) at (6.75, -0.375) {};
		\node [style=none] (59) at (10.25, -0.375) {};
		\node [style=none] (60) at (10.25, 3.5) {};
		\node [style=wire label, color=red] (61) at (11, 3.25) {$\coloneqq D'$};
		\node [style=none] (100) at (1.25, 2.5) {};
		\node [style=none] (101) at (1.25, 1.5) {};
		\node [style=none] (102) at (4.25, 1.5) {};
		\node [style=none] (103) at (4.25, 2.5) {};
		\node [style=none] (104) at (2.75, 2) {$D$};
		\node [style=none] (105) at (1.75, 1.5) {};
		\node [style=none] (106) at (3.75, 1.5) {};
		\node [style=none] (107) at (1.75, 0.75) {};
		\node [style=none] (108) at (3.75, 0.75) {};
		\node [style=wire label] (109) at (1.5, 1.125) {$1$};
		\node [style=wire label] (110) at (4, 1.125) {$1$};
		\node [style=none] (111) at (2.75, 1.125) {$\ldots$};
		\node [style=none] (112) at (1.5, 0.75) {};
		\node [style=none] (113) at (4, 0.75) {};
		\node [style=none] (114) at (2, 0) {};
		\node [style=none] (115) at (3.5, 0) {};
		\node [style=none] (116) at (2.75, 0) {};
		\node [style=none] (122) at (2, -1) {};
		\node [style=none] (123) at (3.5, -1) {};
		\node [style=none] (124) at (2.75, -1) {};
		\node [style=wire label] (125) at (3, -0.5) {$1$};
		\node [style=none] (129) at (1.75, -2.5) {};
		\node [style=none] (130) at (3.75, -2.5) {};
		\node [style=none] (131) at (1.75, -1.75) {};
		\node [style=none] (132) at (3.75, -1.75) {};
		\node [style=wire label] (133) at (1.5, -2.125) {$1$};
		\node [style=wire label] (134) at (4, -2.125) {$1$};
		\node [style=none] (135) at (2.75, -2.125) {$\ldots$};
		\node [style=none] (136) at (1.5, -1.75) {};
		\node [style=none] (137) at (4, -1.75) {};
		\node [style=none, label={above:\axref{S2}}, label={below:\axref{S1}}] (138) at (5.5, 0) {$=$};
		\node [style=none] (139) at (7, 3.25) {};
		\node [style=none] (140) at (7, 2.25) {};
		\node [style=none] (141) at (10, 2.25) {};
		\node [style=none] (142) at (10, 3.25) {};
		\node [style=none] (143) at (8.5, 2.75) {$D$};
		\node [style=none] (144) at (7.5, 2.25) {};
		\node [style=none] (145) at (9.5, 2.25) {};
		\node [style=none] (146) at (7.5, 1.5) {};
		\node [style=none] (147) at (9.5, 1.5) {};
		\node [style=wire label] (148) at (7.25, 1.875) {$1$};
		\node [style=wire label] (149) at (9.75, 1.875) {$1$};
		\node [style=none] (150) at (8.5, 1.875) {$\ldots$};
		\node [style=none] (151) at (7.25, 1.5) {};
		\node [style=none] (152) at (9.75, 1.5) {};
		\node [style=none] (153) at (7.75, 0.75) {};
		\node [style=none] (154) at (9.25, 0.75) {};
		\node [style=none] (155) at (8.5, 0.75) {};
		\node [style=none] (156) at (7.75, -2) {};
		\node [style=none] (157) at (9.25, -2) {};
		\node [style=none] (158) at (8.5, -2) {};
		\node [style=none] (160) at (7.5, -3.5) {};
		\node [style=none] (161) at (9.5, -3.5) {};
		\node [style=none] (162) at (7.5, -2.75) {};
		\node [style=none] (163) at (9.5, -2.75) {};
		\node [style=wire label] (164) at (7.25, -3.125) {$1$};
		\node [style=wire label] (165) at (9.75, -3.125) {$1$};
		\node [style=none] (166) at (8.5, -3.125) {$\ldots$};
		\node [style=none] (167) at (7.25, -2.75) {};
		\node [style=none] (168) at (9.75, -2.75) {};
		\node [style=wire label] (169) at (8.75, 0.5) {$1$};
		\node [style=gn] (170) at (8.5, 0) {};
		\node [style=wire label] (171) at (8.75, -0.75) {$2$};
		\node [style=gn] (172) at (8.5, -1.25) {};
		\node [style=wire label] (173) at (8.75, -1.75) {$1$};
		\node [style=none] (179) at (11.5, 0) {$=$};
		\node [style=none] (180) at (13, 2) {};
		\node [style=none] (181) at (13, 1) {};
		\node [style=none] (182) at (16, 1) {};
		\node [style=none] (183) at (16, 2) {};
		\node [style=none] (184) at (14.5, 1.5) {$D'$};
		\node [style=none] (185) at (14.5, 1) {};
		\node [style=none] (197) at (13.75, -0.75) {};
		\node [style=none] (198) at (15.25, -0.75) {};
		\node [style=none] (199) at (14.5, -0.75) {};
		\node [style=none] (200) at (13.5, -2.25) {};
		\node [style=none] (201) at (15.5, -2.25) {};
		\node [style=none] (202) at (13.5, -1.5) {};
		\node [style=none] (203) at (15.5, -1.5) {};
		\node [style=wire label] (204) at (13.25, -1.875) {$1$};
		\node [style=wire label] (205) at (15.75, -1.875) {$1$};
		\node [style=none] (206) at (14.5, -1.875) {$\ldots$};
		\node [style=none] (207) at (13.25, -1.5) {};
		\node [style=none] (208) at (15.75, -1.5) {};
		\node [style=wire label] (211) at (14.75, 0.5) {$2$};
		\node [style=gn] (212) at (14.5, 0) {};
		\node [style=wire label] (213) at (14.75, -0.5) {$1$};
		\node [style=none, label={above:\lemref{nonscalarcompletetm}}] (214) at (17.25, 0) {$=$};
		\node [style=none] (215) at (18.75, 2) {};
		\node [style=none] (216) at (18.75, 1) {};
		\node [style=none] (217) at (21.75, 1) {};
		\node [style=none] (218) at (21.75, 2) {};
		\node [style=none] (219) at (20.25, 1.5) {$N(D')$};
		\node [style=none] (220) at (20.25, 1) {};
		\node [style=none] (221) at (19.5, -0.75) {};
		\node [style=none] (222) at (21, -0.75) {};
		\node [style=none] (223) at (20.25, -0.75) {};
		\node [style=none] (224) at (19.25, -2.25) {};
		\node [style=none] (225) at (21.25, -2.25) {};
		\node [style=none] (226) at (19.25, -1.5) {};
		\node [style=none] (227) at (21.25, -1.5) {};
		\node [style=wire label] (228) at (19, -1.875) {$1$};
		\node [style=wire label] (229) at (21.5, -1.875) {$1$};
		\node [style=none] (230) at (20.25, -1.875) {$\ldots$};
		\node [style=none] (231) at (19, -1.5) {};
		\node [style=none] (232) at (21.5, -1.5) {};
		\node [style=wire label] (233) at (20.5, 0.5) {$2$};
		\node [style=gn] (234) at (20.25, 0) {};
		\node [style=wire label] (235) at (20.5, -0.5) {$1$};
	\end{pgfonlayer}
	\begin{pgfonlayer}{edgelayer}
		\draw (0.center) to (1.center);
		\draw (1.center) to (2.center);
		\draw (2.center) to (3.center);
		\draw (3.center) to (0.center);
		\draw (7.center) to (5.center);
		\draw (6.center) to (8.center);
		\draw [style=braceedge] (14.center) to (13.center);
		\draw [style=dotsedge, color=red] (59.center)
			 to (58.center)
			 to (57.center)
			 to (60.center)
			 to cycle;
		\draw (100.center) to (101.center);
		\draw (101.center) to (102.center);
		\draw (102.center) to (103.center);
		\draw (103.center) to (100.center);
		\draw (107.center) to (105.center);
		\draw (106.center) to (108.center);
		\draw (112.center)
			 to (114.center)
			 to (115.center)
			 to (113.center)
			 to cycle;
		\draw (123.center)
			 to (122.center)
			 to (136.center)
			 to (137.center)
			 to cycle;
		\draw (124.center) to (116.center);
		\draw (131.center) to (129.center);
		\draw (130.center) to (132.center);
		\draw [style=dotsedge] (137.center) to (136.center);
		\draw (139.center) to (140.center);
		\draw (140.center) to (141.center);
		\draw (141.center) to (142.center);
		\draw (142.center) to (139.center);
		\draw (146.center) to (144.center);
		\draw (145.center) to (147.center);
		\draw (151.center)
			 to (153.center)
			 to (154.center)
			 to (152.center)
			 to cycle;
		\draw (157.center)
			 to (156.center)
			 to (167.center)
			 to (168.center)
			 to cycle;
		\draw (158.center) to (155.center);
		\draw (162.center) to (160.center);
		\draw (161.center) to (163.center);
		\draw [style=dotsedge] (168.center) to (167.center);
		\draw (180.center) to (181.center);
		\draw (181.center) to (182.center);
		\draw (182.center) to (183.center);
		\draw (183.center) to (180.center);
		\draw (198.center) to (197.center);
		\draw (197.center) to (207.center);
		\draw (207.center) to (208.center);
		\draw (208.center) to (198.center);
		\draw (202.center) to (200.center);
		\draw (201.center) to (203.center);
		\draw [style=dotsedge] (208.center) to (207.center);
		\draw (185.center) to (199.center);
		\draw (215.center) to (216.center);
		\draw (216.center) to (217.center);
		\draw (217.center) to (218.center);
		\draw (218.center) to (215.center);
		\draw (222.center) to (221.center);
		\draw (221.center) to (231.center);
		\draw (231.center) to (232.center);
		\draw (232.center) to (222.center);
		\draw (226.center) to (224.center);
		\draw (225.center) to (227.center);
		\draw [style=dotsedge] (232.center) to (231.center);
		\draw (220.center) to (223.center);
	\end{pgfonlayer}
\end{tikzpicture}

  \]
  \[
    \begin{tikzpicture}
	\begin{pgfonlayer}{nodelayer}
		\node [style=none] (0) at (0, 0.025) {$=$};
		\node [style=none] (1) at (2.5, -0.725) {};
		\node [style=none] (2) at (4, -0.725) {};
		\node [style=none] (3) at (3.25, -0.725) {};
		\node [style=none] (4) at (2.25, -2.225) {};
		\node [style=none] (5) at (4.25, -2.225) {};
		\node [style=none] (6) at (2.25, -1.475) {};
		\node [style=none] (7) at (4.25, -1.475) {};
		\node [style=wire label] (8) at (2, -1.85) {$1$};
		\node [style=wire label] (9) at (4.5, -1.85) {$1$};
		\node [style=none] (10) at (3.25, -1.85) {$\ldots$};
		\node [style=none] (11) at (2, -1.475) {};
		\node [style=none] (12) at (4.5, -1.475) {};
		\node [style=wire label] (13) at (3.5, 0.525) {$2$};
		\node [style=gn] (14) at (3.25, 0.025) {};
		\node [style=wire label] (15) at (3.5, -0.475) {$1$};
		\node [style=gbox] (16) at (3.25, 1.275) {$0$};
		\node [style={rn_phase}] (17) at (2.25, 3.525) {$K_{1}$};
		\node [style=uw] (18) at (2.25, 2.275) {};
		\node [style=gbox] (19) at (1.25, 1.275) {$a_0$};
		\node [style=wire label] (20) at (2.575, 2.775) {$2$};
		\node [style=none, label={above:\lemref{zeros-red}}] (21) at (5.75, 0.025) {$=$};
		\node [style=none] (22) at (8.5, -0.975) {};
		\node [style=none] (23) at (10, -0.975) {};
		\node [style=none] (24) at (9.25, -0.975) {};
		\node [style=none] (25) at (8.25, -2.475) {};
		\node [style=none] (26) at (10.25, -2.475) {};
		\node [style=none] (27) at (8.25, -1.725) {};
		\node [style=none] (28) at (10.25, -1.725) {};
		\node [style=wire label] (29) at (8, -2.1) {$1$};
		\node [style=wire label] (30) at (10.5, -2.1) {$1$};
		\node [style=none] (31) at (9.25, -2.1) {$\ldots$};
		\node [style=none] (32) at (8, -1.725) {};
		\node [style=none] (33) at (10.5, -1.725) {};
		\node [style=wire label] (34) at (9.5, 0.275) {$2$};
		\node [style=gn] (35) at (9.25, -0.225) {};
		\node [style=wire label] (36) at (9.5, -0.725) {$1$};
		\node [style=rn] (37) at (9.25, 1.525) {};
		\node [style={rn_phase}] (38) at (8.25, 3.775) {$K_{1}$};
		\node [style=uw] (39) at (8.25, 2.525) {};
		\node [style=gbox] (40) at (7.25, 1.525) {$a_0$};
		\node [style=wire label] (41) at (8.575, 3.025) {$2$};
		\node [style=rn] (42) at (9.25, 0.775) {};
		\node [style=none, label={above:\lemref{wunitlm}}, label={below:\axref{K0}}] (43) at (11.5, 0.025) {$=$};
		\node [style=none] (44) at (18.5, -1.475) {};
		\node [style=none] (45) at (20.5, -1.475) {};
		\node [style=wire label] (46) at (18.25, -1.1) {$1$};
		\node [style=wire label] (47) at (20.75, -1.1) {$1$};
		\node [style=none] (48) at (19.5, -1.1) {$\ldots$};
		\node [style={rn_phase}] (49) at (19.5, 1.775) {$K_{1}$};
		\node [style=gbox] (50) at (19.5, 0.275) {$a_0$};
		\node [style=wire label] (51) at (19.825, 1.025) {$2$};
		\node [style=rn] (52) at (18.5, -0.475) {};
		\node [style=rn] (53) at (20.5, -0.475) {};
		\node [style=none] (54) at (13.75, -0.475) {};
		\node [style=none] (55) at (15.25, -0.475) {};
		\node [style=none] (56) at (14.5, -0.475) {};
		\node [style=none] (57) at (13.5, -1.975) {};
		\node [style=none] (58) at (15.5, -1.975) {};
		\node [style=none] (59) at (13.5, -1.225) {};
		\node [style=none] (60) at (15.5, -1.225) {};
		\node [style=wire label] (61) at (13.25, -1.6) {$1$};
		\node [style=wire label] (62) at (15.75, -1.6) {$1$};
		\node [style=none] (63) at (14.5, -1.6) {$\ldots$};
		\node [style=none] (64) at (13.25, -1.225) {};
		\node [style=none] (65) at (15.75, -1.225) {};
		\node [style=wire label] (66) at (14.75, -0.225) {$1$};
		\node [style=rn] (67) at (14.5, 1.275) {};
		\node [style={rn_phase}] (68) at (13.5, 1.775) {$K_{1}$};
		\node [style=gbox] (69) at (13.5, 0.275) {$a_0$};
		\node [style=wire label] (70) at (13.825, 1.025) {$2$};
		\node [style=none, label={above:\lemref{dim-merge-x}}] (71) at (17, 0.025) {$=$};
	\end{pgfonlayer}
	\begin{pgfonlayer}{edgelayer}
		\draw (2.center) to (1.center);
		\draw (1.center) to (11.center);
		\draw (11.center) to (12.center);
		\draw (12.center) to (2.center);
		\draw (6.center) to (4.center);
		\draw (5.center) to (7.center);
		\draw [style=dotsedge] (12.center) to (11.center);
		\draw (17) to (18);
		\draw [in=105, out=-30, looseness=0.75] (18) to (16);
		\draw [in=-150, out=75, looseness=0.75] (19) to (18);
		\draw (16) to (14);
		\draw (14) to (3.center);
		\draw (23.center) to (22.center);
		\draw (22.center) to (32.center);
		\draw (32.center) to (33.center);
		\draw (33.center) to (23.center);
		\draw (27.center) to (25.center);
		\draw (26.center) to (28.center);
		\draw [style=dotsedge] (33.center) to (32.center);
		\draw (35) to (24.center);
		\draw (38) to (39);
		\draw [in=105, out=-30, looseness=0.75] (39) to (37);
		\draw [in=-150, out=75, looseness=0.75] (40) to (39);
		\draw (42) to (35);
		\draw (52) to (44.center);
		\draw (53) to (45.center);
		\draw (50) to (49);
		\draw (55.center) to (54.center);
		\draw (54.center) to (64.center);
		\draw (64.center) to (65.center);
		\draw (65.center) to (55.center);
		\draw (59.center) to (57.center);
		\draw (58.center) to (60.center);
		\draw [style=dotsedge] (65.center) to (64.center);
		\draw (69) to (68);
		\draw (67) to (56.center);
	\end{pgfonlayer}
\end{tikzpicture}

  \]
  \qedhere
\end{proof}

\begin{restatable}[Z box]{lem}{zboxlm}
  When $d > 1$, then the following holds:
  \label{lem:zbox}
  \[
    \input{qufit-generators/zboxshort.tikz}
  \]
\end{restatable}
\begin{proof}
  This follows from \zxw{Lemma 2} and \autoref{lem:qudit-qufit-equivalence}.
\end{proof}

\begin{restatable}[W node]{lem}{wnodelm}
  \label{lem:wnode}
  When $d > 1$, then the following holds:
  \begin{align*}
    &\input{qufit-generators/wshortqufit-lhs.tikz}
    \scol{\quad=\quad}\dcol{\\&=\quad}
    \input{qufit-generators/wshortqufit-rhs.tikz}
  \end{align*}
\end{restatable}
\begin{proof}
  Follows from \zxw{Lemma 3} and \autoref{lem:qudit-qufit-equivalence}.
\end{proof}

\begin{restatable}[Hadamard box]{lem}{hadlm}
  When $d > 1$, then the following holds:
  \label{lem:hadamard}
  \[
    \dcol{\rotatebox[origin=c]{90}{\input{qufit-generators/hadamardqfit.tikz}}}
    \scol{\input{qufit-generators/hadamardqfit.tikz}}
  \]
\end{restatable}
\Description{Lemma that the Hadamard gate equals its normal form}
\begin{proof}
  Similarly, this follows from \zxw{Lemma 4} and \autoref{lem:qudit-qufit-equivalence}.
\end{proof}

As the proofs of the following three lemmas are nontrivial, they can be found in the appendix.
Notably, the proofs of \autoref{lem:partialtrace} and \autoref{lem:tensor} rely on the same intuition as their equivalent in~\cite{poorCompletenessArbitraryFinite2023}.

\begin{restatable}[Dimension splitter]{lem}{dimsplitterlm}
  \label{lem:dimsplitter}
    When $d > 1$, then the following holds:
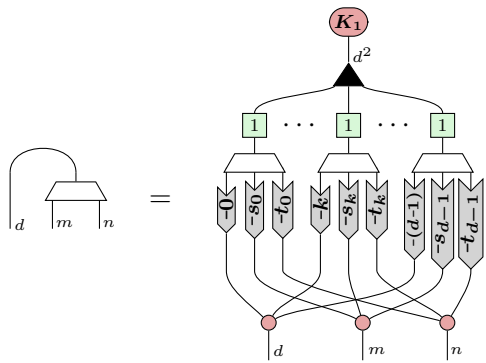
  \[
    \input{qufit-generators/dimspliterv2.tikz}
  \]
  where $d=mn, k=s_kn+t_k, 0 \leq k \leq d-1$.
\end{restatable}

\begin{restatable}[Partial trace]{lem}{partialtracelm}
  \label{lem:partialtrace}
  For $\mathbf{m} = (m_0, \cdots, m_{s - 1})$ where $\prod_{i = 0}^{s - 1} > 1$, the following holds:
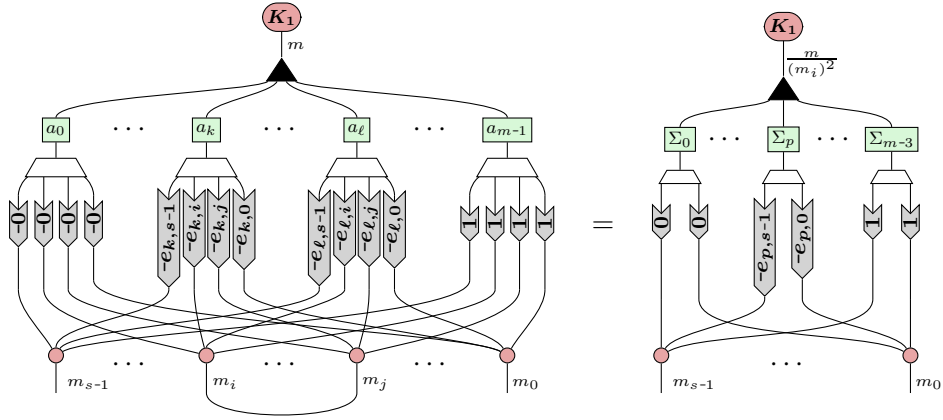
  \begin{align*}
    &\input{figures/qufit-trace/trace-lhs.tikz}
    \scol{\quad=\quad}\dcol{\\&\hspace{2cm}=\quad}
    \input{figures/qufit-trace/trace-rhs.tikz}
  \end{align*}
  where $m_i = m_j$ and $\Sigma_k$ corresponds to the elements of the partial trace over $s$ and $t$ indices.
  That is, $\Sigma_k$ is the sum of the boxes $a_{k_0}, \cdots ,a_{k_{d \minu 1}}$ with such multiplier connections that satisfy $e_{x,k_0} = \cdots = e_{x,k_{d \minu 1}}$ for all $x \in \{0,\, \ldots,\, d - 1\} \setminus \{s, t\}$ and $e_{s,k_{y}} = e_{t,k_y} = y$ for $0 \leq y \leq d - 1$.
\end{restatable}

\begin{restatable}[Tensor product]{lem}{tensorlm}
  \label{lem:tensor}
  For $\mathbf{m} = (m_0, \cdots, m_{s - 1})$ where $\prod_{i = 0}^{s - 1} m_i > 1$ and $\mathbf{n} = (n_0, \cdots, n_{t - 1})$ where $\prod_{i = 0}^{t - 1} n_i > 1$, the following holds:
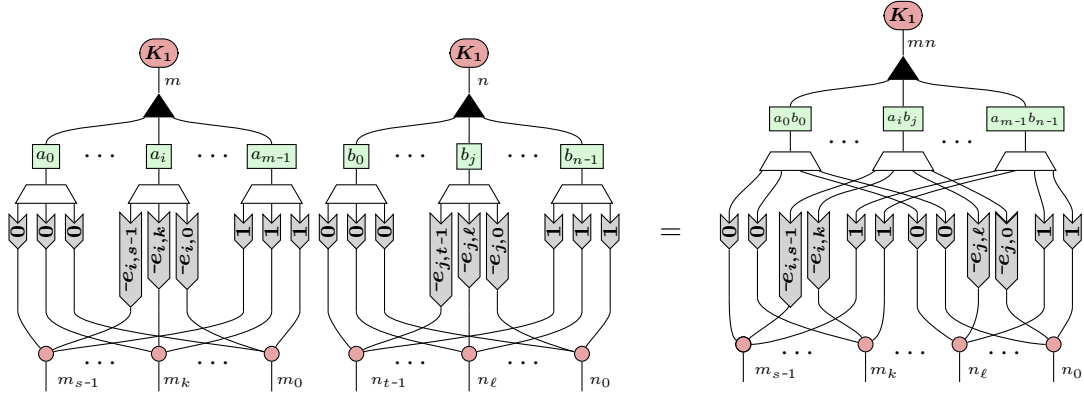
  \begin{align*}
    &\input{figures/qufit-tensor/tensor-lhs.tikz}
    \scol{\quad=\quad}\dcol{\\&\hspace{2cm}=\quad}
    \input{figures/qufit-tensor/tensor-rhs.tikz}
  \end{align*}
\end{restatable}


\section{Applications of qufinite ZXW calculus}\label{sec:apply}

By developing a graphical language for finite-dimensional quantum theory, we can begin to apply the power of diagrammatic reasoning to a wide range of problems in quantum physics and quantum information theory.
In this section, we discuss several potential applications of the qufinite ZXW calculus.

\subsection{Quantum chemistry}

\subsubsection{Spin networks}

The concept of spin networks was described by Roger Penrose as a combinatorial approach to space-time~\cite{penroseAngularMomentumApproach1971}.
It has found applications in quantum chemistry, specifically for the angular momentum coupling problem~\cite{brinkAngularMomentum1994}, and in the theory of quantum gravity~\cite{rovelliSpinNetworks1995}.
The key part of the spin network theory is the irreducible representations of $SU(2)$, which can be directly constructed from the symmetrizer~\cite{martin-dussaudPrimerGroupTheory2019}:
\[
\mathcal{S}_n: (\mathbb{C}^2)^{\otimes n} \longrightarrow  (\mathbb{C}^2)^{\otimes n}, v_1\otimes \cdots \otimes v_n \mapsto \frac{1}{n!}\Sigma_{\sigma \in \mathfrak{S}_n}v_{\sigma_1}\otimes \cdots \otimes v_{\sigma_n},
\]
where $\mathfrak{S}_n$ is the $n$-element permutation group.

East et al.~\cite{eastAKLTStatesZXDiagrams2022, eastSpinnetworksZXcalculus2022} showed that the qubit ZX~calculus can be used to represent and perform calculations on the spin networks.
Since they use a qubit calculus, reasoning about higher dimensional spin representations and couplings between different spin systems requires an encoding in multiple qubits.
On the other hand, the qufinite ZXW calculus offers a platform to directly reason about the coupling of spin systems of arbitrary dimensions.
For example, the following diagram represents the symmetrizer of spin-$\frac{n}{2}$ for any $n\in \mathbb{N}$: 
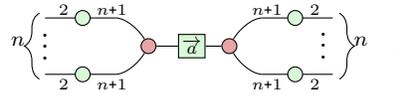
\[
    \begin{tikzpicture}
	\begin{pgfonlayer}{nodelayer}
		\node [style=none] (40) at (3.75, 0.75) {};
		\node [style=gbox] (44) at (0, 0) {$\overrightarrow{a}$};
		\node [style=none] (48) at (1.875, 0.75) {};
		\node [style=none] (51) at (3.5, 0.25) {$\vdots$};
		\node [style=rn] (52) at (1, 0) {};
		\node [style=none] (53) at (1.875, 0.75) {};
		\node [style=none] (54) at (1.875, -0.75) {};
		\node [style=gn] (55) at (2.75, 0.75) {};
		\node [style=none] (56) at (3.75, -0.75) {};
		\node [style=none] (58) at (1.875, -0.75) {};
		\node [style=gn] (60) at (2.75, -0.75) {};
		\node [style=none] (61) at (-3.9, 0.75) {};
		\node [style=wire label] (62) at (-3.4, 0.975) {$2$};
		\node [style=none] (63) at (-2.025, 0.75) {};
		\node [style=none] (65) at (-3.9, 0.225) {$\vdots$};
		\node [style=rn] (66) at (-1.15, 0) {};
		\node [style=none] (67) at (-2.025, 0.75) {};
		\node [style=none] (68) at (-2.025, -0.75) {};
		\node [style=gn] (69) at (-2.9, 0.75) {};
		\node [style=none] (70) at (-3.9, -0.75) {};
		\node [style=wire label] (71) at (-3.4, -1.025) {$2$};
		\node [style=none] (72) at (-2.025, -0.75) {};
		\node [style=gn] (74) at (-2.9, -0.75) {};
		\node [style=wire label] (77) at (-2.125, 0.975) {$n \plus 1$};
		\node [style=wire label] (78) at (-2.125, -1.025) {$n \plus 1$};
		\node [style=wire label] (79) at (3.275, -1.025) {$2$};
		\node [style=wire label] (80) at (2, -1.025) {$n \plus 1$};
		\node [style=wire label] (81) at (3.275, 0.975) {$2$};
		\node [style=wire label] (82) at (2, 0.975) {$n \plus 1$};
		\node [style=none] (83) at (-4.25, -0.875) {};
		\node [style=none] (84) at (-4.25, 0.875) {};
		\node [style=none] (85) at (-4.625, 0.125) {{\scriptsize $n$}};
		\node [style=none] (86) at (4.125, 0.875) {};
		\node [style=none] (87) at (4.125, -0.875) {};
		\node [style=none] (88) at (4.5, 0.125) {{\scriptsize $n$}};
	\end{pgfonlayer}
	\begin{pgfonlayer}{edgelayer}
		\draw [in=60, out=180, looseness=0.75] (53.center) to (52);
		\draw [in=180, out=-60, looseness=0.75] (52) to (54.center);
		\draw (40.center) to (48.center);
		\draw (56.center) to (58.center);
		\draw [in=120, out=0, looseness=0.75] (67.center) to (66);
		\draw [in=0, out=-120, looseness=0.75] (66) to (68.center);
		\draw (61.center) to (63.center);
		\draw (70.center) to (72.center);
		\draw (52) to (44);
		\draw (44) to (66);
		\draw [style=braceedge] (83.center) to (84.center);
		\draw [style=braceedge] (86.center) to (87.center);
	\end{pgfonlayer}
\end{tikzpicture}
\quad ,
\]
\begin{flalign*}
    \text{where }\,
    \overrightarrow{a} = \left(
        \frac{1}{\binom{n}{1}}, \cdots, \frac{1}{\binom{n}{k}},  \cdots, \frac{1}{\binom{n}{n}}
        \right).&&
\end{flalign*}
Hence, we expect that the qufinite ZXW calculus would serve as a valuable tool for reasoning about spin networks.

\begin{rem}
  Since this paper first appeared online, this particular idea has been explored and used for providing a diagrammatic langauge for spin SU(2)~\cite{wangPenroseTensor2025} and loop quantum gravity~\cite{priestleyFiniteDimensionalZXCalculus2025}.
\end{rem}

\subsubsection{Interacting mixed dimensional quantum systems}

In quantum chemistry, we often observe interactions between different types of particles, such as the interaction between electrons and photons.
Different molecules have varying degrees of freedom and are modeled by different Hilbert spaces.
We can extend the qudit ZXW framework for reasoning about Hamiltonians in quantum chemistry~\cite{shaikhHowSum2023,defeliceLightMatterInteractionZXW2023} to mixed-dimensional quantum systems using the qufinite ZXW calculus.
This allows us to reason about interacting mixed dimensional systems in a diagrammatic way.
A natural example of this is the Hamiltonian of the Jaynes-Cummings model, which describes the interaction between a two-level atom and a photon, that can be represented as follows:
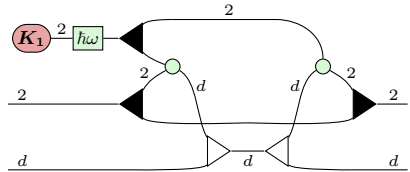
\[
    \input{applications/JaynesCummings.tikz}
\]
where the white triangle can be represented as a ZXW diagram (see~\cite{defeliceLightMatterInteractionZXW2023}) with the following interpretation:\tikzstyle{lsplit}=[shape=isosceles triangle, isosceles triangle stretches=true, fill=white, draw=black, minimum width=0.4cm, minimum height=3mm, inner sep=1pt, shape border rotate=180]
\tikzstyle{lsplit}=[shape=isosceles triangle, isosceles triangle stretches=true, fill=white, draw=black, minimum width=0.4cm, minimum height=3mm, inner sep=1pt, shape border rotate=180]
\[
  \input{applications/split.tikz}
  \quad \xmapsto{\interp{\cdot}} \quad
  \ket{n} \mapsto \sum_{k=0}^n \binom{n}{k}^{\frac{1}{2}} \ket{k} \ket{n-k}
\]

\subsection{Quantum programming, algorithms and circuits}

\subsubsection{Quantum programs as qufinite ZXW diagrams}

Scalable ZX calculus~\cite{caretteSZXCalculusScalable2019, caretteColoredPropsLarge2020} has been shown to be useful for quantum programming~\cite{borgnaEncodingHighlevelQuantum2023}, where the main power comes from the generator \emph{gatherer} and its transpose \emph{divider}.
These two generators allow dealing with the register of qubits.
It turns out that the divider and gatherer are a special case of the dimension splitter and  merger respectively, where the dimensions are restricted to powers of 2:
\tikzstyle{divide}=[regular polygon, regular polygon sides=3, shape border rotate=90, draw=black, fill={zx_grey}, inner sep=1.5pt, tikzit category=scal, rounded corners=0.8mm]
\tikzstyle{gather}=[fill={zx_grey}, draw=black, tikzit category=scal, rounded corners=0.8mm, regular polygon, regular polygon sides=3, shape border rotate=-90, inner sep=1.5pt]
\begin{align*}
    &\input{applications/scalable-gather.tikz} : \C^{2^n} \otimes \C^{2^m} \mapsto \C^{2^{n + m}}
    \dcol{\\&\hspace{2cm}}\qquad \cong \qquad
    \input{applications/qufinite-merge.tikz} : \C^{2^n} \otimes \C^{2^m} \mapsto \C^{2^{n+m}}
\end{align*}
Therefore, our calculus has the ability to express similar ideas.
However, with the completeness of the qufinite ZXW calculus, we can expect further improvements in this area of research.

\subsubsection{High-level language for quantum algorithms}

To do anything interesting with a quantum computer, we need good quantum algorithms.
However, designing quantum algorithms is a challenging task.
Having a high-level graphical language could help us to understand the structure of quantum algorithms and thus design new ones.
Qufinite ZXW calculus is a promising candidate for such a language.
For example, the circuit for Quantum Fourier Transform can be nicely represented by a qufinite ZXW diagram as follows:
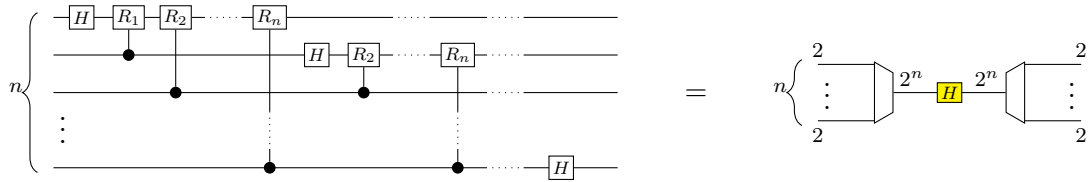
\begin{align*}
    &\input{applications/QFT-circuit.tikz}
    \dcol{\\&\hspace{2.5cm}}\qquad = \qquad
    \input{applications/QFT-diagramv2.tikz}
\end{align*}
Therefore, it would be interesting to explore if we can describe high-level primitives and routines in the calculus. Furthermore, we can also find the relation between other high-level reasoning frameworks such as sum-over-paths~\cite{amyLargescaleFunctionalVerification2018} and the qufinite ZXW calculus.

\subsubsection{Mixed-dimensional quantum computing}

Quantum circuits with mixed-dimensional qudits have advantages in certain applications.
Qufinite ZXW calculus allows us to leverage the strength of diagrammatic compilation and optimization techniques to the mixed-dimensional case.
For example, the following diagram implements a CNOT gate between a qubit and a $d$-level qudit.
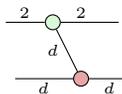
\[
    \input{applications/mixed-dim-cnot.tikz}
\]
As a simple example, we can show that this CNOT, when applied $d$ times, is equivalent to the identity:
\begin{align*}
    \input{applications/mixed-dim-cnots-0.tikz}\quad
    &\input{applications/mixed-dim-cnots-1.tikz}\dcol{\\}
    &\input{applications/mixed-dim-cnots-2.tikz}
\end{align*}

\section{Conclusion and further work}\label{sec:future}

In this paper, we introduced the qufinite ZXW calculus, a graphical language for finite dimensional quantum theory.
Subsequently, we established a normal form for arbitrary tensors.
Via this normal form, we proved the completeness of the qufinite ZXW calculus, demonstrating that any equality between two tensors can be derived through diagrammatic rewriting using the complete set of rules.
A major consequence of this result is the equivalence between the qufinite ZXW calculus and the category of finite-dimensional Hilbert spaces.
Since the first version of the current paper, the completeness in $\FHilb$ has been extended to both the ZX calculus~\cite{poorZXcalculusComplete2024} and the ZW calculus~\cite{devismeMinimalityFiniteDimensionalZWCalculi2024}.

Several promising directions for future work are already presented in \autoref{sec:apply}.
Another compelling direction of interest is the application of qufinite ZXW calculus to tensor network contraction.
With any tensor now being representable within qufinite ZXW calculus and each equality between two tensors derivable through diagrammatic rewriting,
it would be interesting to explore techniques and strategies for tensor contraction based on rewriting.
Traditional tensor network evaluation heavily depends on tensor contraction and tensor decomposition methods (e.g.\@ SVD), which can be computationally expensive, especially for high-dimensional tensors.
In contrast, rewriting relies on pattern matching, which can be employed to rewrite the topology of the tensor network, potentially reducing bottlenecks during contraction.
An example of this technique is presented in~\cite{camSpeedingQuantumCircuits2023} using the ZX calculus;
the additional flexibility provided by the qufinite ZXW calculus may offer further advantages in optimizing tensor network operations.

It is important to note that, up to this point in our work, the exploration of ruleset minimality has been a preliminary pursuit.
By questioning the necessity of each rule and attempting to derive them from others, we anticipate that --- similar to the achievements in the qubit ZX-calculus~\cite{vilmartNearMinimalAxiomatisationZXCalculus2019} --- the complete rule-set can be further streamlined.
This refinement would capture the essential interactions more concisely and conveniently.

\section*{Acknowledgements}

We would like to thank Pablo Andres-Martinez, Bob Coecke, Alexander Cowtan, Giovanni de Felice, Amar Hadzihasanovic, Mark Koch, Sam Staton, John van de Wetering, and Lia Yeh
for their detailed feedback and their numerous suggestions for improvement of the paper.
We particularly thank John van de Wetering for his valuable discussion with us on the diagrammatic representation of symmetrizers.
RS is supported by the Clarendon Fund Scholarship.

\bibliographystyle{alphaurl}
\bibliography{preamble/references-bibtex}

\appendix

\allowdisplaybreaks
\setlength{\jot}{20pt}

\section{Lemmas}\label{set:lemmas}

\subsection{Lemmas for the qudit setting}\label{subsec:qudit-lemmas}

\begin{lem}
  \label{kjgalm}\cite{wangQufiniteZXcalculusUnified2022}
  \[
    \input{lemmas/kjga.tikz}
  \]
  where $\overrightarrow{a}=(a_1,\cdots, a_{d-1})$, $j \in \{ 1,\cdots, d-1\}.$
\end{lem}
\begin{proof}
  Same as~\zxw{Lemma 7}.
\end{proof}



\begin{lem}
  \label{zeroemptyditlm}\cite{wangQufiniteZXcalculusUnified2022}
  \input{lemmas/zeroemptydit.tikz}
\end{lem}
\begin{proof}
  Same as~\zxw{Lemma 9}.
\end{proof}


\begin{lem}
  \label{k1zbox1}
  \[
    \input{lemmas/k1-zbox1.tikz}
  \]
\end{lem}
\begin{proof}
  \[
    \input{lemmas/k1-zbox1-pf.tikz}
  \]
  \qedhere
\end{proof}



\begin{lem}
  \label{s4lm}\cite{wangQufiniteZXcalculusUnified2022}
  \begin{gather}
    \input{lemmas/redspider0pfusedit2.tikz}
    \tag{S4}\label{rule:S4}
  \end{gather}
\end{lem}
\begin{proof}
  Same as~\zxw{Lemma 13}.
\end{proof}




\begin{lem}
  \label{zid}
  \input{lemmas/z-id.tikz}
\end{lem}
\begin{proof}
  Same as~\cite[Lemma 22]{devismeMinimalityFiniteDimensionalZWCalculi2024}
\end{proof}

\begin{lem}
  \label{ww-projector}
  \input{lemmas/w-w-algebra-projector.tikz}
\end{lem}
\begin{proof}
  Same as~\cite[Lemma 23]{devismeMinimalityFiniteDimensionalZWCalculi2024}
\end{proof}

\begin{lem}
  \label{dboxsquarelm}
  \input{lemmas/dboxsquare.tikz}
\end{lem}
\begin{proof}
  Same as~\zxw{Lemma 16}.
\end{proof}


\begin{lem}
  \label{dboxgdotlm}
  \[
    \input{lemmas/dboxgdot.tikz}
  \]
\end{lem}
\begin{proof}
  Same as~\zxw{Lemma 17}.
\end{proof}




\begin{lem}
  \label{dboxslidegrnlm}
  \[
    \input{lemmas/dboxslidegrn-2.tikz}
  \]
\end{lem}
\begin{proof}
  Same as~\zxw{Lemma 20}.
\end{proof}

\begin{lem}
  \label{multiplierpushlm}
  Suppose $x \in \{0, \dotsc, d - 1 \}$.
  Then
  \[
    \input{lemmas/multipliers/multiplier-push.tikz} \qquad
    \qquad
    \input{lemmas/multipliers/multiplier-push-dual.tikz}
  \]
\end{lem}
\begin{proof}
  Same as~\zxw{Lemma 31}.
\end{proof}


\begin{lem}
  \label{multipliermultlm}
  Suppose $x, y \in \{0, \dotsc, d - 1 \}$.
  Then
  \[
    \input{lemmas/multipliers/multiplier-mult.tikz}
  \]
\end{lem}
\begin{proof}
  Same as~\zxw{Lemma 32}.
\end{proof}

\begin{lem}
  \label{multiplieraddlm}
  \[
    \input{lemmas/multipliers/multiplier-add.tikz}
  \]
\end{lem}
\begin{proof}
  Same as~\zxw{Lemma 33}.
\end{proof}


\begin{lem}
  \label{dualisermultipliermultlm}
  Suppose $x \in \{0, \dotsc, d - 1 \}$.
  Then
  \[
    \input{lemmas/multipliers/dualiser-multiplier-mult.tikz}
  \]
\end{lem}
\begin{proof}
  Same as~\zxw{Lemma 34}.
\end{proof}

\begin{lem}
  \label{dboxgcopylm}
  \[
    \input{lemmas/dboxgcopy.tikz}
  \]
\end{lem}
\begin{proof}
  Same as~\zxw{Lemma 35}.
\end{proof}

\begin{lem}
  \label{multiplier-kj}
  \[
    \input{lemmas/multipliers/multiplier-kj.tikz}
  \]
\end{lem}
\begin{proof}
  This follows from Rules~\eqref{rule:Mu},~\eqref{rule:K0}, and~\eqref{rule:S4}.
\end{proof}

\begin{lem}
  \label{xtransposelm}
  \[
    \input{lemmas/xtranspose.tikz}
  \]
\end{lem}
\begin{proof}
  Same as~\zxw{Lemma 36}.
\end{proof}

\begin{lem}
  \label{2pinkfusionlm}
  \[
    \input{qudit-nf-trace/lem1.tikz}
  \]
\end{lem}
\begin{proof}
  Same as~\zxw{Lemma 48}.
\end{proof}


\begin{lem}
  \label{wdecomplm}
  \[
    \input{lemmas/wdecomp.tikz}
  \]
\end{lem}
\begin{proof}
  Same as~\zxw{Lemma 37} with the definition of the triangle expanded.
\end{proof}

\begin{lem}
  \label{wunitlm}
  \[
    \input{lemmas/wunit.tikz}
  \]
\end{lem}
\begin{proof}
  Same as~\zxw{Lemma 38}.
\end{proof}

\begin{lem}
  \label{wk1}
  \[
    \input{lemmas/wk1.tikz}
  \]
\end{lem}
\begin{proof}
  \[
    \input{lemmas/wk1-pf.tikz}
  \]
  \qedhere
\end{proof}


\begin{lem}
  \label{wtransposelm}
  \[
    \input{lemmas/wtranspose.tikz}
  \]
\end{lem}
\begin{proof}
  Same as~\zxw{Lemma 44}.
\end{proof}


%
%


\begin{lem}
  \label{lem:qudit-zboxnf}
  \[
    \input{qudit-nf-generators/zboxshort.tikz}
  \]
\end{lem}
\begin{proof}
  Same as~\zxw{Lemma 1}
\end{proof}

\begin{lem}
  \label{lem:qudit-wnodenf}
  \[
    \input{qudit-nf-generators/wshort.tikz}
  \]
\end{lem}
\begin{proof}
  Same as~\zxw{Lemma 2}
\end{proof}
%
%

\begin{lem}
  \label{multiplier-zbox}
  \[
    \input{lemmas/multipliers/multiplier-zbox.tikz}
  \]
\end{lem}
\begin{proof}
  \[
    \input{lemmas/multipliers/multiplier-zbox-pf.tikz}
  \]
  \qedhere
\end{proof}

\begin{lem}
  \label{dualiser-wnode}
  \[
    \input{lemmas/dualiser-wnode.tikz}
  \]
\end{lem}
\begin{proof}
  By map-state duality, we only have to show the following:
  \begin{align*}
    \input{lemmas/dualiser-wnode-pf-0.tikz}\quad
    & \input{lemmas/dualiser-wnode-pf-1.tikz} \\
    & \input{lemmas/dualiser-wnode-pf-2.tikz}
  \end{align*}
  where the penultimate equality follows from~\zxw{Lemma 1} which proves that
  any qudit normal form can be rearranged.
\end{proof}

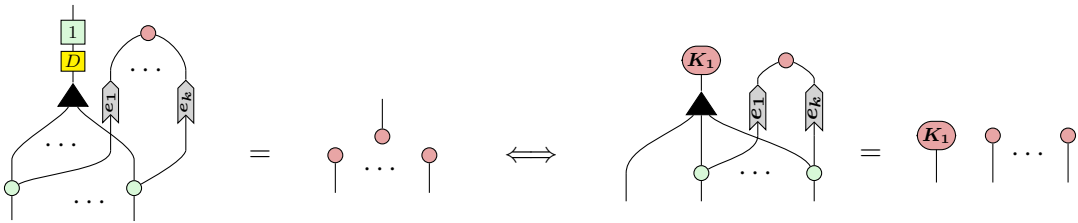
\begin{lem}
  \label{kzequivalencelm}
  \[
     \input{qufit-axioms/kzqudit.tikz}
     \qquad \Longleftrightarrow \qquad
     \input{axioms/k1zstate.tikz}
  \]
\end{lem}
\begin{proof}
  $\Longrightarrow$: Given the left-hand side equation (LHS), we prove that:
  \begin{align*}
    \input{lemmas/kzeuivnecv2-pf-0.tikz}\quad
    & \input{lemmas/kzeuivnecv2-pf-1.tikz} \\
    & \input{lemmas/kzeuivnecv2-pf-2.tikz}
  \end{align*}
  $\Longleftarrow$: Given the right-hand side equation (RHS), we prove that:
    \begin{align*}
    \input{lemmas/kzeuivsuf-pf-0.tikz}\quad
    & \input{lemmas/kzeuivsuf-pf-1.tikz} \\
    & \input{lemmas/kzeuivsuf-pf-2.tikz}
  \end{align*}
  \qedhere
\end{proof}

\begin{lem}
  \begin{equation}
    \input{lemmas/triangleocopydit.tikz}
    \tag{Bs0}\label{rule:Bs0}
  \end{equation}
  \label{lem:Bs0}
\end{lem}
\begin{proof}
  \[
    \input{lemmas/triangleocopydit-pf.tikz}
  \]
  \qedhere
\end{proof}

\begin{lem}
  \begin{equation}
    \begin{tikzpicture}
	\begin{pgfonlayer}{nodelayer}
		\node [style=gbox] (0) at (-1, 0.5) {$0$};
		\node [style=rn] (1) at (1, 0.5) {};
		\node [style=none] (2) at (-1, -0.5) {};
		\node [style=none] (3) at (1, -0.5) {};
		\node [style=none] (4) at (0, 0) {=};
	\end{pgfonlayer}
	\begin{pgfonlayer}{edgelayer}
		\draw (2.center) to (0);
		\draw (1) to (3.center);
	\end{pgfonlayer}
\end{tikzpicture}

    \tag{Zer}\label{rule:Zer}
  \end{equation}
  \label{lem:Zer}
\end{lem}
\begin{proof}
  \[
    \input{lemmas/zerotoreddit0-pf.tikz}
  \]
  \qedhere
\end{proof}

\begin{lem}
  \begin{equation*}
    \input{lemmas/zeros-red.tikz}
  \end{equation*}
  \label{zeros-red}
\end{lem}
\begin{proof}
  \[
    \input{lemmas/zeros-red-pf.tikz}
  \]
  \qedhere
\end{proof}

\begin{lem}
  \begin{equation}
    \input{lemmas/pimultiplecpdit.tikz}
    \tag{K1}\label{rule:K1}
  \end{equation}
  \label{lem:K1}
\end{lem}
\begin{proof}
  \[
    \input{lemmas/pimultiplecpdit-pf.tikz}
  \]
  \qedhere
\end{proof}

\begin{lem}
  \begin{equation}
    \input{lemmas/z-push-v.tikz}
    \tag{ZV}\label{rule:ZV}
  \end{equation}\label{lem:ZV}
where $\overrightarrow{a_{d-1}} = \left(a_{d-1}, a_{d-1}, \dotsc, a_{d-1}\right)$.
\end{lem}
\begin{proof}
  \[
    \input{lemmas/z-push-v-pf.tikz}
  \]
  \qedhere
\end{proof}
\begin{lem}
  \begin{equation}
    \input{lemmas/old-va.tikz}
    \tag{VA}
  \end{equation}
  \label{lem:VA}
\end{lem}
\begin{proof}
  \[
    \input{lemmas/old-va-pf.tikz}
  \]
  \qedhere
\end{proof}
\begin{lem}
  \begin{equation}
    \input{lemmas/v-push-w.tikz}
    \tag{VW}\label{rule:VW}
  \end{equation}
  \label{lem:VW}
\end{lem}
\begin{proof}
  \[
    \input{lemmas/v-push-w-pf.tikz}
  \]
  \qedhere
\end{proof}
\begin{lem}
  \begin{equation}
    \input{lemmas/had-decomposition.tikz}
  \end{equation}
  \label{lem:HD}
\end{lem}
\begin{proof}
  \[
\input{lemmas/had-decomposition-pf.tikz}
  \]
  \qedhere
\end{proof}

\quditImplied*
\begin{proof}
  We iterate through each Axioms of Ref.~\cite{poorCompletenessArbitraryFinite2023} and show that they can be derived from the rules of $\ZXWf$.

  \paragraph{The ZX-part of the rules are implied as follows:}
  \begin{description}
    \item[(S1):] implied by \textaxref{S1} and \autoref{eq:z-flexsymmetry}.
    \item[(S2):] implied by \textaxref{S2} and \autoref{zid}
    \item[(D1):] implied by \textaxref{D1}
    \item[(Ept):] same as \textaxref{Ept}
    \item[(B2):] same as \textaxref{B2}
    \item[(K0):] implied by \textaxref{K0} and \autoref{kjgalm} and \autoref{zeroemptyditlm}
    \item[(K1):] see \autoref{lem:K1}
    \item[(K2):] same as \textaxref{K2}
    \item[(Zer):] see \autoref{lem:Zer}
    \item[(P1):] implied by \textaxrefs{P1}{D1} and \autoref{zid}
    \item[(H1):] can be proven using only \textaxref{D1}
  \end{description}

  \paragraph{The ZW-part of the rules are implied as follows:}
  \begin{description}
    \item[(Pcy):] implied by \autoref{lem:Pcy}
    \item[(Sym):] follows from \autoref{rule:Sym}
    \item[(BZW):] implied by \textaxref{BZW}
    \item[(Ad):] same as \textaxref{AD}
    \item[(Aso):] follows from \textaxref{WF}
    \item[(WW):] same as \textaxref{WW}
  \end{description}

  \paragraph{The ZXW-part of the rules are implied as follows:}
  \begin{description}
    \item[(Bs0):] see \autoref{lem:Bs0}
    \item[(Bsj):] same as \textaxref{Bsj}
    \item[(TA):] implied by \textaxref{TA}
    \item[(HD):] see \autoref{lem:HD}
    \item[(ZV):] see \autoref{lem:ZV}
    \item[(VA):] see \autoref{lem:VA}
    \item[(VW):] see \autoref{lem:VW}
    \item[(KZ):] consequence of \autoref{lem:KZ} and \autoref{kzequivalencelm}
  \end{description}
\end{proof}

\subsection{Lemmas for the mixed-dimensional setting}

\begin{lem}
  \begin{equation}
    \input{qufit-lemmas/qufit-phasecopy.tikz}
    \tag{Pcy}\label{rule:Pcy}
  \end{equation}
  \label{lem:Pcy}
\end{lem}
\begin{proof}
  \[
    \input{qufit-lemmas/qufit-phasecopy-pf.tikz}
  \]
  \qedhere
\end{proof}

\begin{lem}
  \label{trialgbzwlm}
  Suppose $m\geq 2$.
  Then
  \[
    \input{lemmas/trialgbzw.tikz}
  \]
\end{lem}
\begin{proof}
  We prove this by induction.
  The base case follows from the Rules~\eqref{rule:TA} and~\eqref{rule:BZW}.
  Then, let us suppose the lemma true for $k - 1$ outputs $(\ast)$.
  The inductive step is proved as follows:
  \[
    \input{lemmas/trialgbzwprf.tikz}
  \]
  \qedhere
\end{proof}

\begin{lem}
  \label{trialgebraallconnectionlm}
  \[
    \input{lemmas/trialgebraallconnection.tikz}
  \]
\end{lem}
\begin{proof}
This follows from the Rule~\eqref{rule:BZW} and \autoref{trialgbzwlm}.
\end{proof}

\begin{lem}
  \label{gerneraltrialgebralm}
  \[
    \input{lemmas/gerneraltrialgebra.tikz}
  \]
\end{lem}
\begin{proof}
  This equality can be derived directly from \autoref{trialgebraallconnectionlm} by splitting the leftmost W spider using ~\eqref{rule:WN}.
\end{proof}

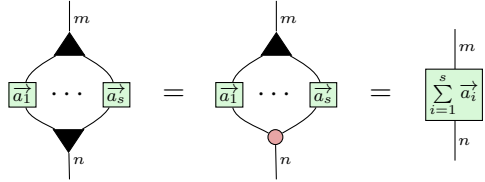
\begin{lem}
  \label{lem:KZ}
  \[
    \input{qufit-lemmas/KZ.tikz}
    \tag{KZ}\label{rule:KZ}
  \]
\end{lem}
\begin{proof}
  \[
    \input{qufit-lemmas/kz-pf.tikz}
  \]
  \qedhere
\end{proof}

\begin{lem}
  \label{2wadditionlm}
  \[
    \input{lemmas/2waddition.tikz}
  \]
\end{lem}
\begin{proof}
  The first equation follows from Rule~\eqref{rule:S1} and \autoref{gerneraltrialgebralm}.
  The second equation can be proved the same way as in ~\zxw{Lemma 42}.
\end{proof}

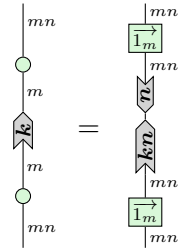
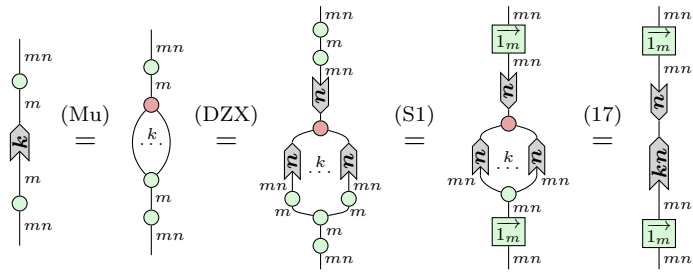
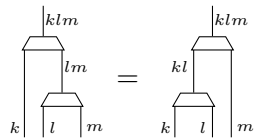
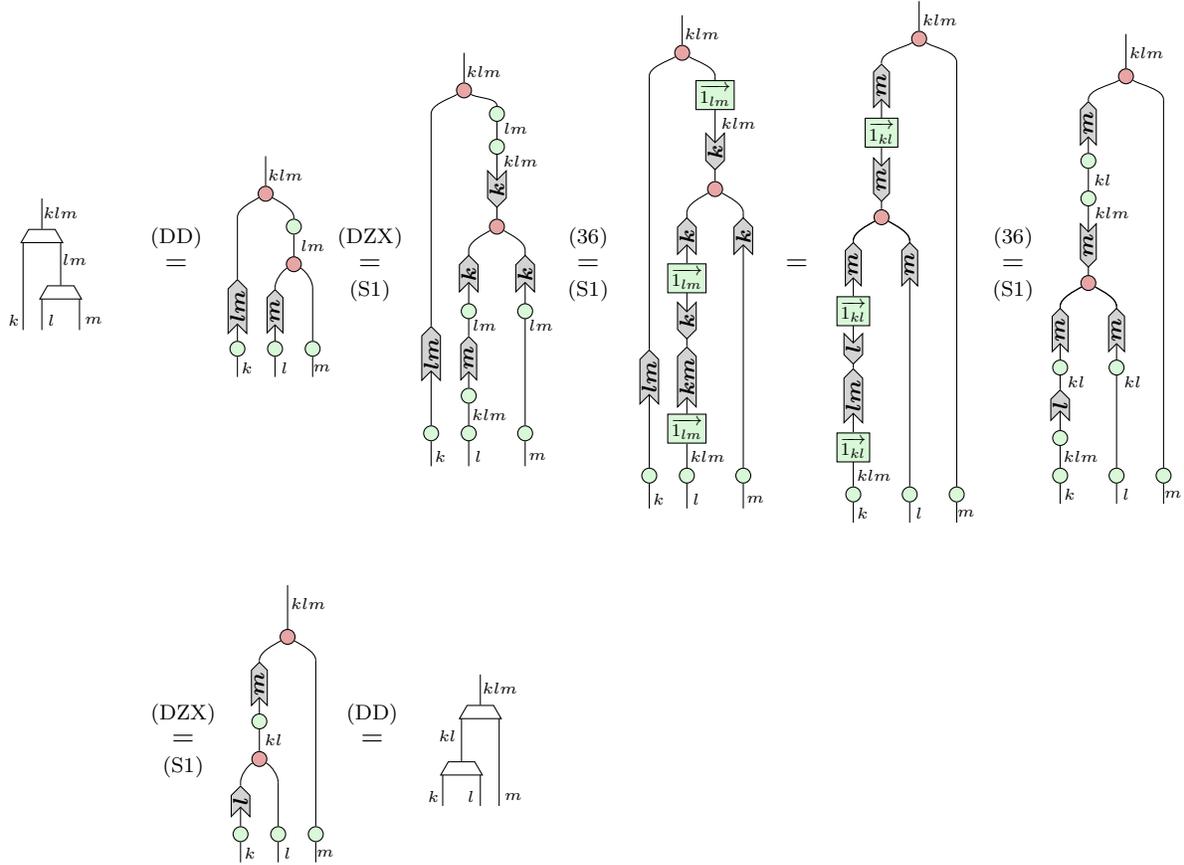
\begin{lem}
  \label{multiplier-zspider}
  \[
    \input{figures/qufit-lemmas/multiplier-zspider.tikz}
  \]
\end{lem}
\begin{proof}
  \[
    \input{figures/qufit-lemmas/multiplier-zspider-pf.tikz}
  \]
  \qedhere
\end{proof}

\begin{lem}
  \label{rule:DAS}
  \[
    \input{qufit-generators/dimboxassociative.tikz}
  \]
\end{lem}
\begin{proof}
  \begin{align*}
    \input{figures/qufit-lemmas/aso-pf-0.tikz}\quad
    & \input{figures/qufit-lemmas/aso-pf-1.tikz} \\
    & \input{figures/qufit-lemmas/aso-pf-2.tikz}
  \end{align*}
  where for the fourth equality, we used the qudit completeness result.
\end{proof}

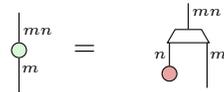
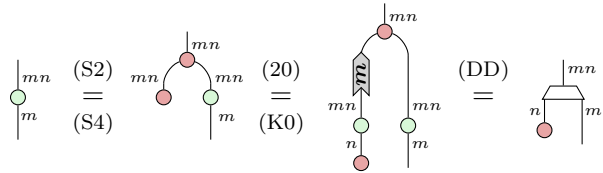
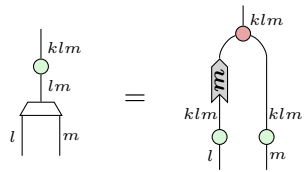
\begin{lem}
  \label{zspider-split-x}
  \[
    \input{qufit-lemmas/zspider-split-x.tikz}
  \]
\end{lem}
\begin{proof}
  \[
    \input{qufit-lemmas/zspider-split-x-pf.tikz}
  \]
  \qedhere
\end{proof}

\begin{lem}
  \label{n-dimbox-decomp-2}
  \[
    \input{qufit-lemmas/n-dimbox-decomp-2.tikz}
  \]
\end{lem}
\begin{proof}
If $k=1$, then it follows directly from Rule~\eqref{rule:DD}. If $k>1$, then
  \[
    \input{qufit-lemmas/n-dimbox-decomp-pf-2.tikz}
  \]
  \qedhere
\end{proof}

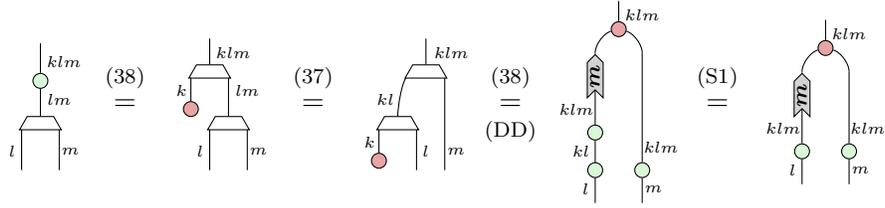
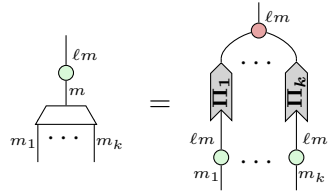
\begin{lem}
  \label{n-dimbox-z-decomp}
  \[
    \input{qufit-lemmas/n-dimbox-z-decomp.tikz}
  \]
  \begin{flalign*}
    \text{where}\quad
    \Pi_j = \prod_{i = j + 1}^{k} m_i \quad
    \text{for } 1 \leq j \leq k - 1,\quad
    \Pi_k \coloneqq 1, \quad
    \text{and}\quad
    m = \prod_{i = 1}^{k} m_i. &&
  \end{flalign*}
\end{lem}
\begin{proof}
  We prove the lemma inductively, in which the base case is \autoref{n-dimbox-decomp-2}.
  Then, let us suppose the lemma true for $k - 1$ outputs $(\ast)$.
  The inductive step is proved as follows:
  \[
    \input{qufit-lemmas/n-dimbox-z-decomp-pf.tikz}
  \]
  \qedhere
\end{proof}

\begin{lem}
  \label{n-dimbox-decomp}
  \[
    \input{qufit-lemmas/n-dimbox-decomp.tikz}
  \]
\end{lem}
\begin{proof}
  This follows from Rule~\eqref{rule:S2} and \autoref{n-dimbox-z-decomp} where we let $l=1$.
\end{proof}

\begin{lem}
  \label{dimboxunitarity-1-qudit-lm}
  \[
    \input{qufit-lemmas/dimboxunitarity-1-qudit-lm.tikz}
  \]
\end{lem}
\begin{proof}
  These two qudit diagrams correspond to the same linear map; hence, the proof follows from the completeness of the qudit ZXW calculus.
\end{proof}
\begin{lem}
  \label{rule:UT1}
  \[
    \input{qufit-generators/dimboxunitarity-1.tikz}
  \]
\end{lem}
\begin{proof}
  \[
    \input{qufit-lemmas/dimboxunitarity-1-pf.tikz}
  \]
  \qedhere
\end{proof}
\begin{lem}
  \label{dim-merge-x}
  \[
    \input{figures/qufit-lemmas/dim-merge-x.tikz}
  \]
  where \quad $0 \leq i \leq m-1,\quad 0 \leq j \leq n-1$.
\end{lem}
\begin{proof}
  \[
    \input{figures/qufit-lemmas/dim-merge-x-pf.tikz}
  \]
  \qedhere
\end{proof}
\begin{lem}
  \label{rule:DC0}
  \[
    \input{figures/qufit-lemmas/dim-split-x0.tikz}
  \]
\end{lem}
\begin{proof}
  \[
    \input{figures/qufit-lemmas/dim-split-x0-pf.tikz}
  \]
  \qedhere
\end{proof}

\begin{lem}
  \label{rule:DC1}
  \[
    \input{figures/qufit-lemmas/dsplit-k1.tikz}
  \]
\end{lem}
\begin{proof}
  \[
    \input{figures/qufit-lemmas/dsplit-k1-pfv2.tikz}
  \]
  \qedhere
\end{proof}

\begin{lem}
  \label{dsplit-zspider}
  \[
    \input{figures/qufit-lemmas/dsplit-zspider.tikz}
  \]
\end{lem}
\begin{proof}
  \[
    \input{figures/qufit-lemmas/dsplit-zspider-pf-new.tikz}
  \]
  where the last equality follows due to the completeness of the qudit ZXW calculus.
\end{proof}

\begin{lem}
  \label{k1splitdiscard}
  \[
    \input{figures/qufit-lemmas/k1-split-discard.tikz}
  \]
\end{lem}
\begin{proof}
  \[
    \input{figures/qufit-lemmas/k1-split-discard-pf.tikz}
  \]
  \qedhere
\end{proof}

\begin{lem}
  \label{lem:MZD}
  \[
    \input{figures/qufit-lemmas/mzd.tikz}
  \]
\end{lem}
\begin{proof}
  \[
    \input{figures/qufit-lemmas/mzd-pf.tikz}
  \]
  where the penultimate step is due to qudit completness.
\end{proof}

\begin{lem}
  \label{lem:id-mixed-z-split}
  \[
    \input{figures/qufit-lemmas/id-z-split.tikz}
  \]
\end{lem}
\begin{proof}
  \[
    \input{figures/qufit-lemmas/id-z-split-pf.tikz}
  \]
  \qedhere
\end{proof}

\begin{lem}
  \label{multiplier-dim-lm-2}
  \[
    \input{figures/qufit-lemmas/multiplier-dim-lm-2.tikz}
  \]
\end{lem}
\begin{proof}
  These two qudit diagrams correspond to the same linear map; hence, the proof follows from the completeness of the qudit ZXW calculus.
\end{proof}
\begin{lem}
  \label{multiplier-dim-lm}
  \[
    \input{figures/qufit-lemmas/multiplier-dim-lm.tikz}
  \]
\end{lem}
\begin{proof}
  \[
    \input{figures/qufit-lemmas/multiplier-dim-lm-pf.tikz}
  \]
  \qedhere
\end{proof}
\begin{lem}
  \label{rule:MS}
  \[
    \input{figures/qufit-lemmas/multiplier-dim.tikz}
  \]
\end{lem}
\begin{proof}
  \[
    \input{figures/qufit-lemmas/multiplier-dim-pf.tikz}
  \]
  \qedhere
\end{proof}

\begin{lem}
  \label{dimbox1-spider-lm}
  \[
    \input{qufit-lemmas/dimbox1-spider-lm.tikz}
  \]
\end{lem}
\begin{proof}
  By \autoref{lem:post-compose-mixed-z}, we only need to show the following:
  \[
    \input{qufit-lemmas/dimbox1-spider-lm-pf.tikz}
  \]
  where the penultimate equality follows from qudit completeness.
\end{proof}

\begin{lem}
  \label{dimbox1-spider}
  \[
    \input{qufit-lemmas/dimbox1-spider.tikz}
  \]
\end{lem}
\begin{proof}
  This follows from \autoref{dboxsquarelm} and \autoref{dimbox1-spider-lm}.
\end{proof}

\begin{lem}
  \label{rule:DPM}
  \[
    \input{figures/qufit-lemmas/dpm.tikz}
  \]
\end{lem}
\begin{proof}
  \[
    \input{figures/qufit-lemmas/dpm-new-pf.tikz}
  \]
  The other equality can be proved similarly.
\end{proof}

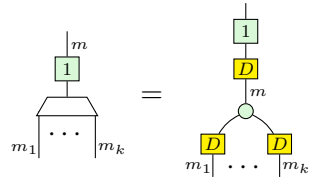
\begin{lem}
  \label{n-dimbox1-spider}
  \[
    \input{qufit-lemmas/n-dimbox1-spider.tikz}
  \]
\end{lem}
\begin{proof}
  We prove the lemma inductively, in which the base case is \autoref{dimbox1-spider}.
  Then, let us suppose the lemma true for $k - 1$ outputs $(\ast)$.
  The inductive step is proved as follows:
  \[
    \input{qufit-lemmas/n-dimbox1-spider-pf.tikz}
  \]
  \qedhere
\end{proof}

\begin{lem}
  \label{rule:BZD}
  \[
    \input{figures/qufit-axioms/ax2.tikz}
  \]
\end{lem}
\begin{proof}
  \[
    \input{figures/qufit-lemmas/BZD-pf.tikz}
  \]
  \qedhere
\end{proof}

\begin{lem}
  \label{rule:UT2}
  \[
    \input{figures/qufit-lemmas/dimboxunitarity-2.tikz}
  \]
\end{lem}
\begin{proof}
  \[
    \input{qufit-lemmas/dimboxunitarity-2-pf.tikz}
  \]
  \qedhere
\end{proof}

\begin{lem}
  \label{DZS-lm-1}
  \[
    \input{figures/qufit-lemmas/ax11-lm-1.tikz}
  \]
\end{lem}
\begin{proof}
  \[
    \input{figures/qufit-lemmas/ax11-lm-1-pf.tikz}
  \]
  \qedhere
\end{proof}

\begin{lem}
  \label{DZS-lm-2}
  \[
    \input{figures/qufit-lemmas/ax11-lm-2.tikz}
  \]
\end{lem}
\begin{proof}
  \[
    \input{figures/qufit-lemmas/ax11-lm-2-pf.tikz}
  \]
  where
  $T_{d \minu 1}=\overbrace{(\underbrace{0,\dotsc, 1}_{d - 1}, \dotsc, 0)}^{d^m - 1}$
  \qedhere
\end{proof}

\begin{lem}
  \label{rule:DZS}
  \[
    \input{figures/qufit-axioms/ax11.tikz}
  \]
\end{lem}
\begin{proof}
  \[
    \input{figures/qufit-lemmas/ax11-pf.tikz}
  \]
  \qedhere
\end{proof}

\begin{lem}
  \label{zspiders-split}
  \[
    \input{figures/qufit-lemmas/zspiders-split.tikz}
  \]
\end{lem}
\begin{proof}
  \[
    \input{figures/qufit-lemmas/zspiders-split-pf.tikz}
  \]
  \qedhere
\end{proof}

\begin{lem}
  \label{zspider-split-discard}
  \[
    \input{figures/qufit-lemmas/zspider-split-discard.tikz}
  \]
\end{lem}
\begin{proof}
  \[
    \input{figures/qufit-lemmas/zspider-split-discard-pf.tikz}
  \]
  \qedhere
\end{proof}


\begin{lem}
  \label{1split-to-zspider}
  \[
    \input{figures/qufit-lemmas/1split-to-zspider.tikz}
  \]
\end{lem}
\begin{proof}
  \[
    \input{figures/qufit-lemmas/1split-to-zspider-pf.tikz}
  \]
  \qedhere
\end{proof}


\begin{lem}
  \label{1split-sym}
  \[
    \input{figures/qufit-lemmas/1split-sym.tikz}
  \]
\end{lem}
\begin{proof}
  This follows from \autoref{dimbox1-spider} and Rule~\eqref{rule:S1}.
\end{proof}




\begin{lem}
  \label{rule:PCZ}
  \[
    \input{figures/qufit-lemmas/pcz.tikz}
  \]
\end{lem}
\begin{proof}
  \[
    \input{figures/qufit-lemmas/pcz-lm-2-pfv2.tikz}
  \]
  \qedhere
\end{proof}


\begin{lem}
  \label{rule:BWD}
  \[
    \input{figures/qufit-lemmas/bwd.tikz}
  \]
\end{lem}
\begin{proof}
  \[
  \input{figures/qufit-lemmas/bwd-lm-1-pf-0v2.tikz}
 \]
  \qedhere
\end{proof}


\begin{lem}
  \label{rule:PCW}
  \[
    \input{qufit-lemmas/w-commute-dsplit-zspider.tikz}
  \]
\end{lem}
\begin{proof}
\[
    \input{qufit-lemmas/w-commute-dsplit-zspider-prf.tikz}
  \]
  \qedhere
\end{proof}



\begin{lem}
  \label{lem:split-mult-add}
  \[
    \input{figures/qufit-trace/lem-1.tikz}
  \]
\end{lem}
\begin{proof}
  \[
    \input{figures/qufit-trace/lem-1-pf.tikz}
  \]
  \qedhere
\end{proof}

\begin{lem}
  \label{lem:dsplit-general-trialg}
  \[
    \input{figures/qufit-lemmas/dsplit-general-trialg.tikz}
  \]
\end{lem}
\begin{proof}
 Firstly,
  \[
    \input{figures/qufit-lemmas/dsplit-general-trialg-pf.tikz}
  \]
  Starting from the third diagram above, the second equality can be proved as follows:
  \[
    \input{figures/qufit-lemmas/dsplit-general-trialg-pf-2.tikz}
  \]
  \qedhere
\end{proof}

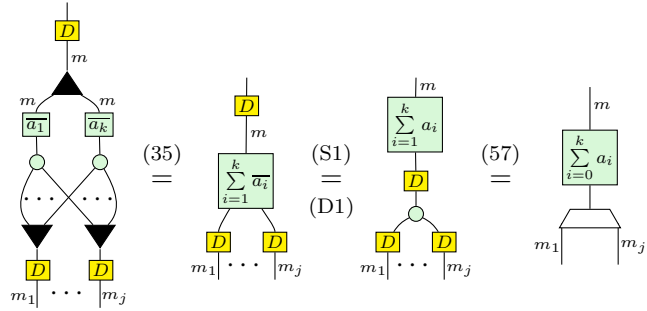
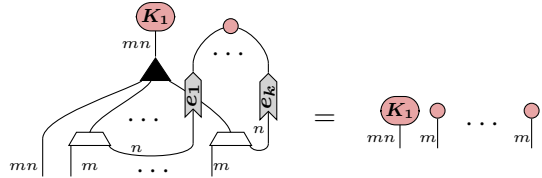
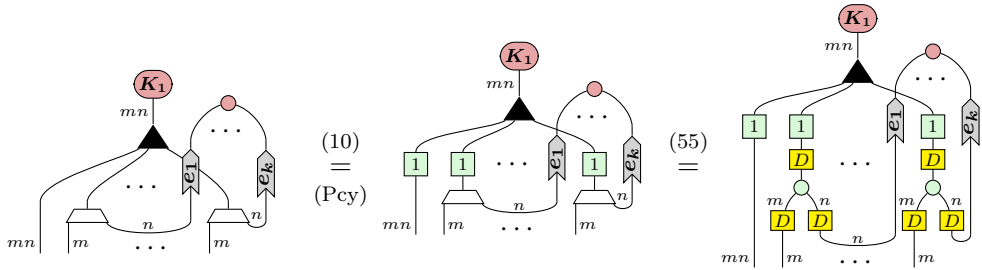
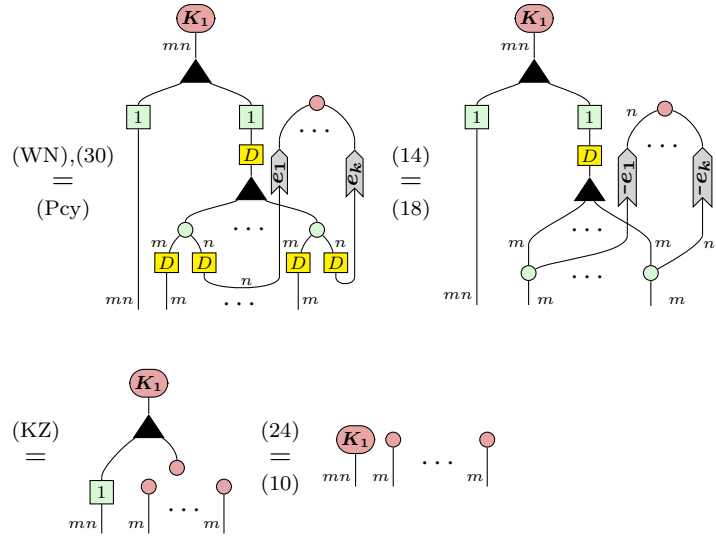
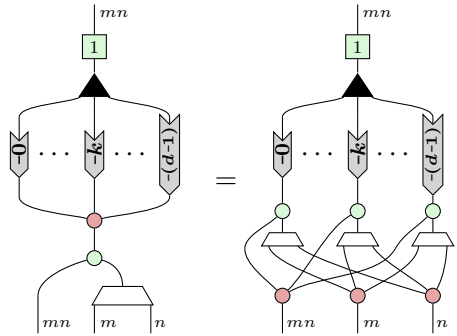
\begin{lem}
  \label{rule:DKZ}
  \[
    \input{figures/qufit-lemmas/dkz.tikz}
  \]
\end{lem}
\begin{proof}
  \begin{align*}
    \input{figures/qufit-lemmas/dkz-pf-0.tikz}\quad
    & \input{figures/qufit-lemmas/dkz-pf-1.tikz} \\
    & \input{figures/qufit-lemmas/dkz-pf-2.tikz} \\
    & \input{figures/qufit-lemmas/dkz-pf-3.tikz}
  \end{align*}
  \qedhere
\end{proof}
\begin{lem}
  \label{rule:DZC}
  \[
    \input{figures/qufit-lemmas/dzc.tikz}
  \]
   where $d=mn$.
\end{lem}
\begin{proof}
  By \autoref{lem:post-compose-mixed-z}, we only need to show the following:
  \begin{align*}
    \input{figures/qufit-lemmas/dzc-pf-0.tikz}\quad
    & \input{figures/qufit-lemmas/dzc-pf-1.tikz} \\
    & \input{figures/qufit-lemmas/dzc-pf-2.tikz}
  \end{align*}
  where the penultimate equality is due to qudit completeness.
\end{proof}

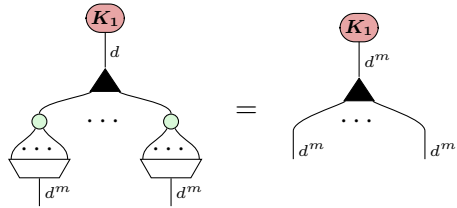
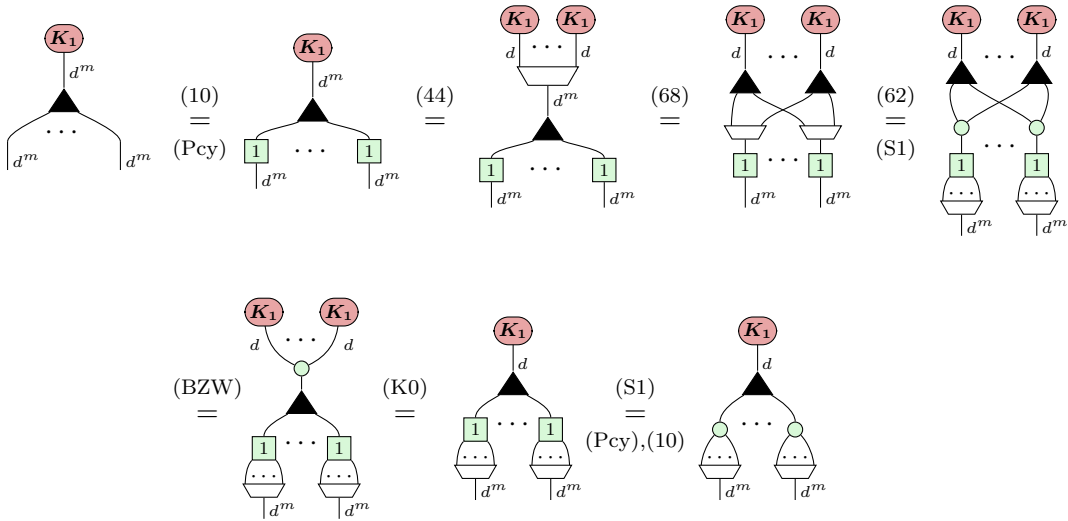
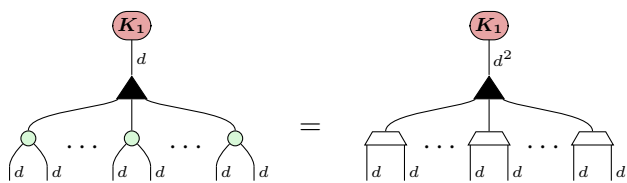
\begin{lem}
  \label{k1dimchange}
  \[
    \input{figures/qufit-lemmas/k1dimchange.tikz}
  \]
\end{lem}
\begin{proof}
  \begin{align*}
    \input{qufit-lemmas/k1dimchange-pf-0.tikz}
    & \input{qufit-lemmas/k1dimchange-pf-1.tikz} \\
    & \input{qufit-lemmas/k1dimchange-pf-2.tikz}
  \end{align*}
  \qedhere
\end{proof}

\begin{lem}
  \label{k1wmultilm}
  \[
    \input{qufit-generators/k1wmultiplierequiv.tikz}
  \]
  which is equivalent to
  \[
    \input{qufit-generators/k1wmultiplier.tikz}
  \]
  where $d=mn$.
\end{lem}
\begin{proof}
  \[
    \input{qufit-generators/k1wmultiplierequivprf.tikz}
  \]
  \qedhere
\end{proof}



\quditqufitequivalence*
\begin{proof}
  \begin{align*}
    \input{qufit-normal-form/qufitnormalformtrans-pf0.tikz}\quad
    & \input{qufit-normal-form/qufitnormalformtrans-pf1.tikz} \\
    & \input{qufit-normal-form/qufitnormalformtrans-pf2.tikz}
  \end{align*}
  \qedhere
\end{proof}

\section{Proof of completeness}

\subsection{Generators}

%
%
%
%

\dimsplitterlm*
\begin{proof}
  \begin{align*}
    \input{qufit-generators/split-pf-0.tikz} \quad
    & \input{qufit-generators/split-pf-1.tikz} \\
    & \input{qufit-generators/split-pf-2.tikz} \\
    & \input{qufit-generators/split-pf-3.tikz}
  \end{align*}
  \qedhere
\end{proof}

\subsection{Partial trace}

\allowdisplaybreaks
\setlength{\jot}{20pt}

\begin{lem}
  \label{lem:tracelm-1}
  \[
    \input{figures/qufit-trace/tracelm-2.tikz}
  \]
  where $\displaystyle \Sigma \coloneqq \sum\limits_{j = 0}^{d \minu 1} \alpha_j$.
\end{lem}
\begin{proof}
  \begin{align*}
    &\input{figures/qufit-trace/tracepf-2-1.tikz}\\
    &\input{figures/qufit-trace/tracepf-2-2.tikz}
  \end{align*}
  \qedhere
\end{proof}

\begin{lem}
  \label{lem:tracelm-2}
  \[
    \input{figures/qufit-trace/tracepf-n-1.tikz}
  \]
\end{lem}
\begin{proof}
  \begin{align*}
    &\input{figures/qufit-trace/tracepf-n-1-pf-1.tikz}\\
    &\input{figures/qufit-trace/tracepf-n-1-pf-2.tikz}
  \end{align*}
  \qedhere
\end{proof}

\partialtracelm*
\begin{proof}
  \begin{align*}
    &\input{figures/qufit-trace/tracepf-1.tikz}\\
    &\input{figures/qufit-trace/tracepf-2.tikz}\\
    &\input{figures/qufit-trace/tracepf-3.tikz}\\
    &\input{figures/qufit-trace/tracepf-4.tikz}
  \end{align*}
  \noindent
  Here, indices can be grouped based on the number of connection to the $m_i \land m_j$ X-spider:
  \begin{enumerate}
    \item Indices that have no connection to the $m_i \land m_j$ X-spider, that is, $e_{k,j} = e_{k,i}$.
    \item Indices that had some non-zero connection to the $m_i \land m_j$ X-spider that is, $e_{k,j} \neq e_{k,i}$.
  \end{enumerate}
  We first show how elements of Group 1 are combined.
  We consider a set of Group 1 indices $k_0, \cdots, k_{m_i \minu 1}$ such that their connection to each X-spider equals.
  That is, for all $0 \leq j \leq s - 1$, the number of connection to the $j$-th X-spider equal from all Z-boxes with index $k_\ell$ for $0 \leq \ell \leq m_i \minu 1$.
  Diagrammatically, this is depicted as follows with $k$ being a representative of $[k_\ell]$:
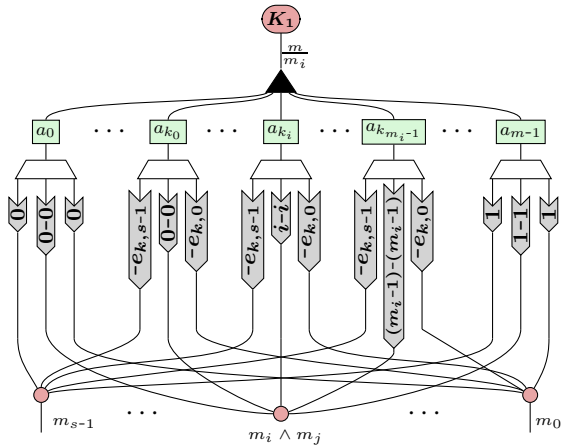
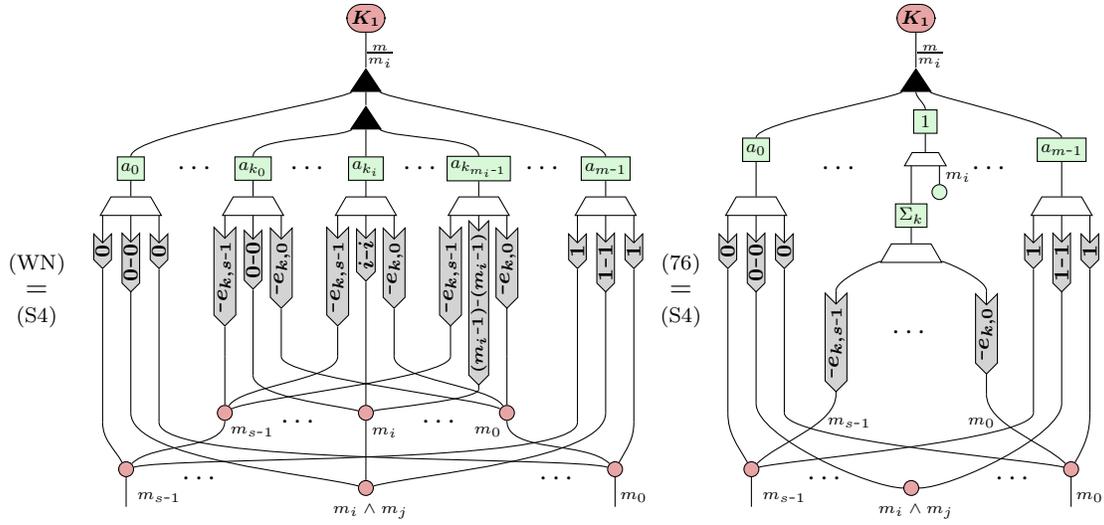
  \begin{align*}
    &\input{figures/qufit-trace/tracepf-3-01.tikz}\\
    &\input{figures/qufit-trace/tracepf-3-02.tikz}
  \end{align*}
  After applying \autoref{lem:tracelm-1} to all $(m-3)$ possible sub-diagrams with equal connections as shown above, indices of Group 2 can be rewritten and eliminated as follows:
  \begin{align*}
    &\input{figures/qufit-trace/tracepf-3-1.tikz}\\
    &\input{figures/qufit-trace/tracepf-3-15.tikz}\\
    &\input{figures/qufit-trace/tracepf-3-2.tikz}
  \end{align*}
  \qedhere
\end{proof}

\subsection{Tensor product}

\allowdisplaybreaks
\setlength{\jot}{20pt}

This section proves \autoref{lem:tensor}, that is, the tensor product of two qufinite normal form can be rewritten into a single normal form.
The proof is split into several sub-lemmas proving small bits of the main lemma.

\begin{lem}
  \label{lem:tensorlm-1}
  \[
    \input{figures/qufit-tensor/tensorlem-1.tikz}
  \]
\end{lem}
\begin{proof}
  \begin{align*}
    &\input{figures/qufit-tensor/tensorpf-1.tikz}\\
    &\input{figures/qufit-tensor/tensorpf-2.tikz}\\
    &\input{figures/qufit-tensor/tensorpf-3.tikz}\\
  \end{align*}
  \qedhere
\end{proof}

\begin{lem}
  \label{lem:tensorlm-2}
  \[
    \input{figures/qufit-tensor/tensorlem-2.tikz}
  \]
\end{lem}
\begin{proof}
  \begin{align*}
    &\input{figures/qufit-tensor/tensorpf-3-2.tikz}\\
    &\input{figures/qufit-tensor/tensorpf-3-3.tikz}\\
    &\input{figures/qufit-tensor/tensorpf-3-4.tikz}\\
    &\input{figures/qufit-tensor/tensorpf-3-5.tikz}\\
  \end{align*}
  \qedhere
\end{proof}

\begin{lem}
  \label{lem:tensorlm-3}
  \[
    \input{figures/qufit-tensor/tensorlem-3.tikz}
  \]
\end{lem}
\begin{proof}
  \begin{align*}
    &\input{figures/qufit-tensor/tensorlem-3-pf-1.tikz}\\
    &\input{figures/qufit-tensor/tensorlem-3-pf-2.tikz}\\
    &\input{figures/qufit-tensor/tensorlem-3-pf-3.tikz}\\
    &\input{figures/qufit-tensor/tensorlem-3-pf-4.tikz}
  \end{align*}
  \qedhere
\end{proof}

\tensorlm*
\begin{proof}
  \begin{align*}
    &\input{figures/qufit-tensor/tensorpf2-1.tikz}\\
    &\input{figures/qufit-tensor/tensorpf2-2.tikz}\\
    &\input{figures/qufit-tensor/tensorpf-10.tikz}\\
    &\input{figures/qufit-tensor/tensorpf-11.tikz}\\
    &\input{figures/qufit-tensor/tensorpf-12.tikz}\\
    &\input{figures/qufit-tensor/tensorpf-13.tikz}
  \end{align*}
  \qedhere
\end{proof}

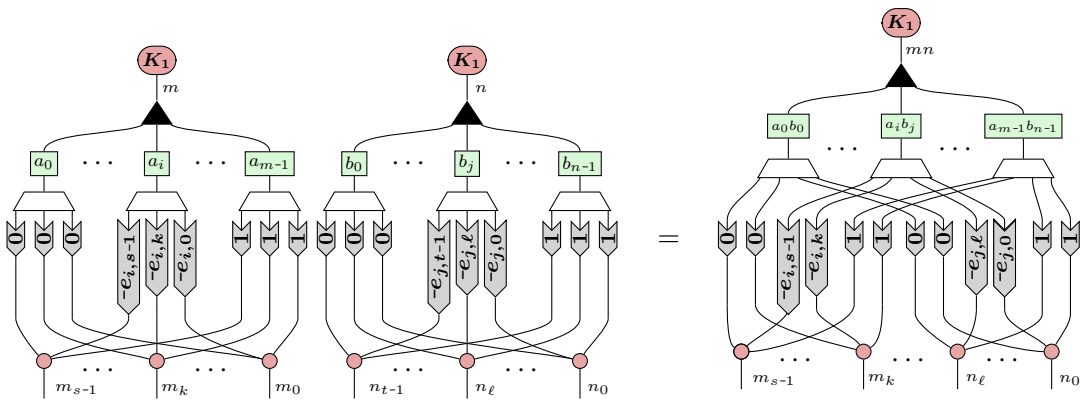
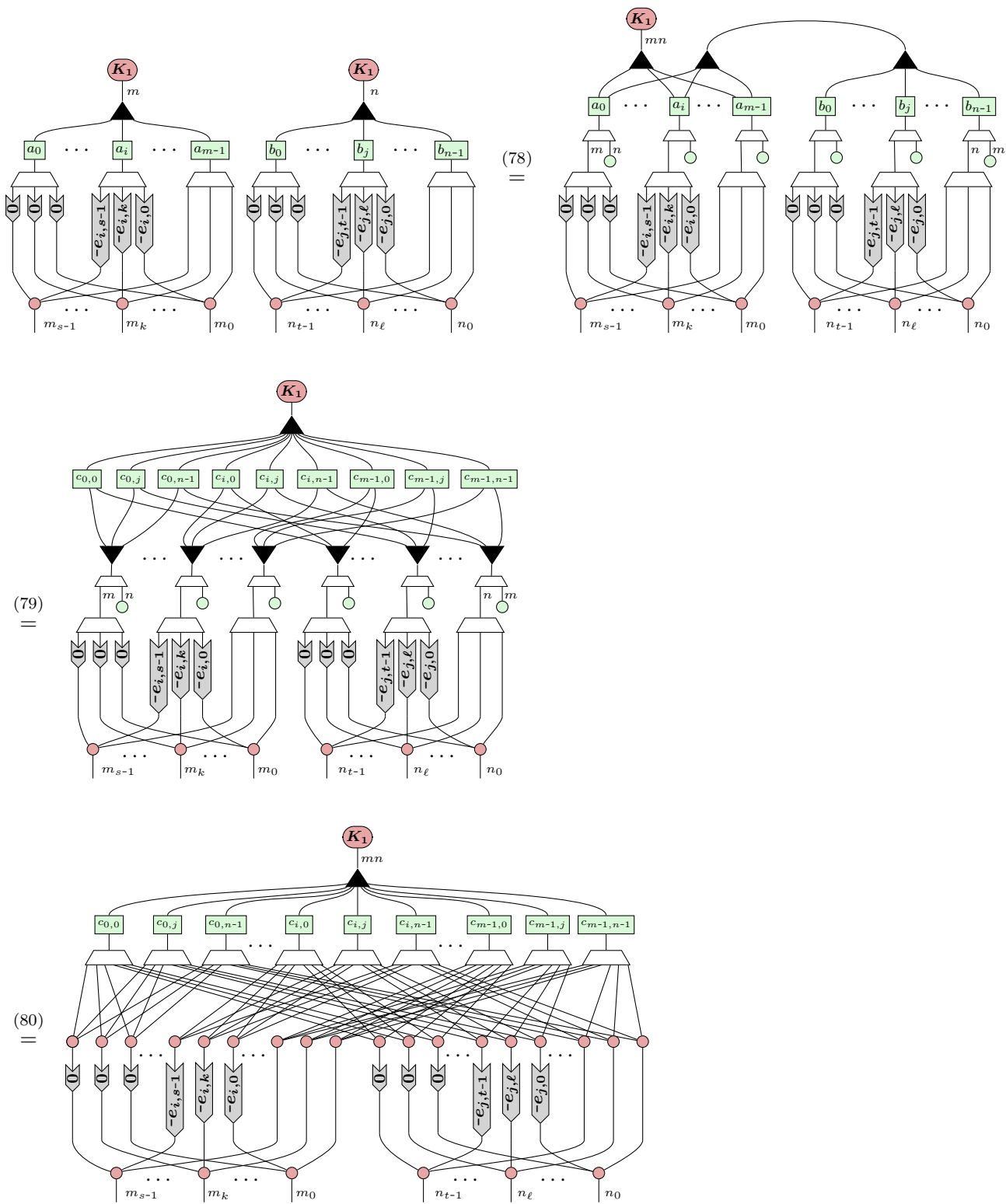
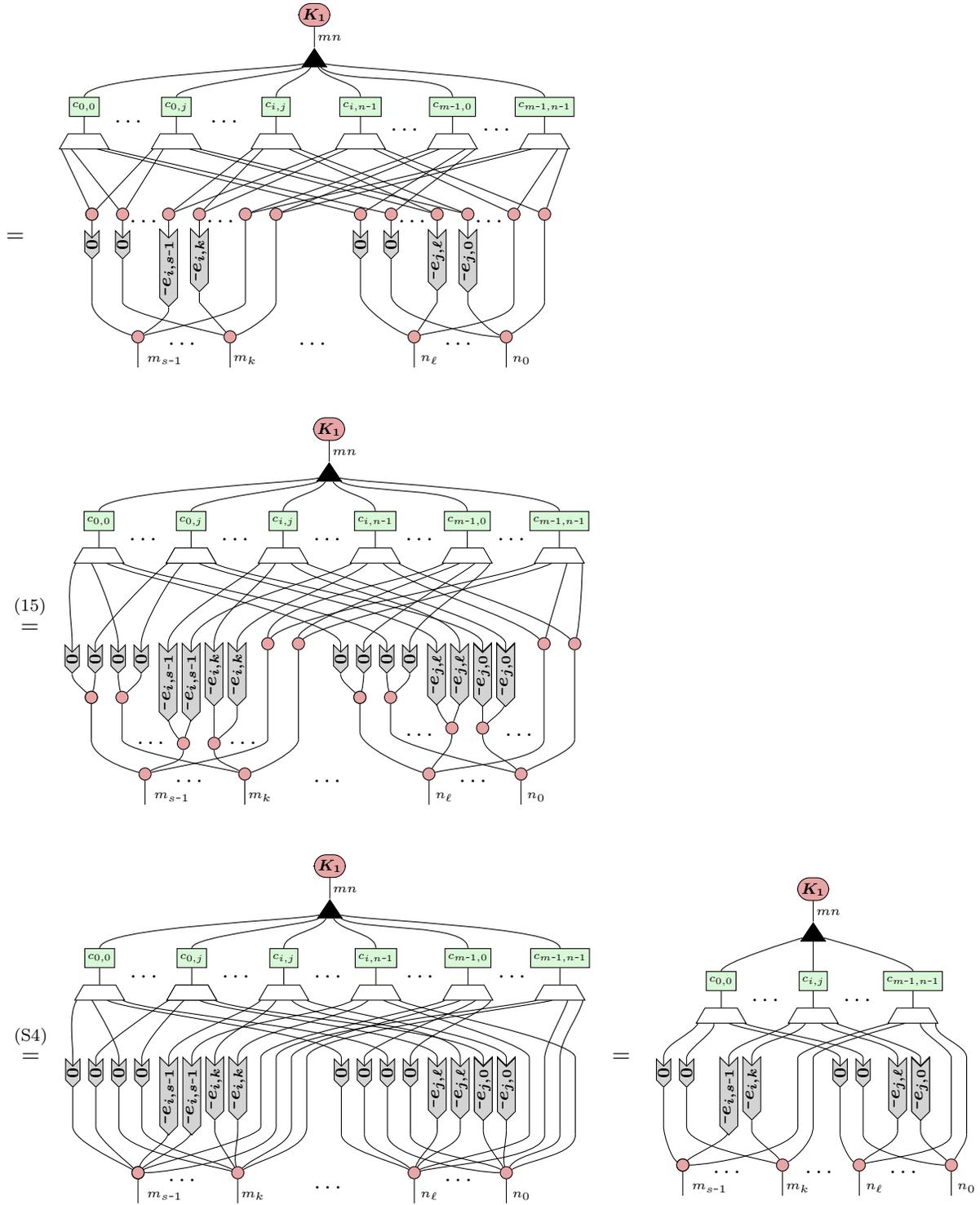
\end{document}